%% file: fscd-lmcs.tex
\documentclass{lmcs}
\pdfoutput=1
\usepackage[utf8]{inputenc}

\usepackage{lastpage}
\lmcsdoi{19}{2}{8}
\lmcsheading{}{\pageref{LastPage}}{}{}%
{Nov.~02,~2021}{May~04,~2023}{}

\pdfoutput=1 


\usepackage{amssymb}
\usepackage{amsthm}

\usepackage{tikz-cd}

\usepackage{microtype}

\usepackage[hypertexnames=false]{hyperref}
\usepackage[capitalise,noabbrev]{cleveref}

\Crefname{defi}{Definition}{Definitions}
\Crefname{thm}{Theorem}{Theorems}
\Crefname{thmC}{Theorem}{Theorems}
\Crefname{cor}{Corollary}{Corollaries}
\Crefname{prop}{Proposition}{Propositions}
\Crefname{propC}{Proposition}{Propositions}
\Crefname{lem}{Lemma}{Lemmas}
\Crefname{lemC}{Lemma}{Lemmas}
\Crefname{rem}{Remark}{Remarks}
\Crefname{exa}{Example}{Examples}
\Crefname{exas}{Examples}{Examples}

\input{mymacros.tex}

\theoremstyle{plain}\newtheorem*{clmnn}{Claim}

\def\eg{{\em e.g.}}
\def\cf{{\em cf.}}
\def\ie{{\em i.e.}}

\begin{document}

\title[On Small Types in Univalent Foundations]%
{On Small Types in Univalent Foundations\rsuper*}
\titlecomment{{\lsuper*}%
  This is a revised and extended version of~\cite{FSCDversion}.}

\author[T.~de~Jong]{Tom de Jong\lmcsorcid{0000-0003-1585-3172}}
\author[M.~H.~Escard\'o]{Mart\'in H\"otzel Escard\'o\lmcsorcid{0000-0002-4091-6334}}
\address{School of Computer Science, University of Birmingham, Birmingham, B15 2TT, UK}
\email{tom.dejong@nottingham.ac.uk, m.escardo@cs.bham.ac.uk}

\begin{abstract}
  We investigate predicative aspects of constructive univalent
  foundations. By predicative and constructive, we respectively mean that we do
  not assume Voevodsky's propositional resizing axioms or excluded middle.
  Our work complements existing work on predicative mathematics by
  exploring what cannot be done predicatively in univalent foundations.
  Our first main result is that nontrivial (directed or bounded) complete posets
  are necessarily large. That is, if such a nontrivial poset is small, then weak
  propositional resizing holds. It is possible to derive full propositional
  resizing if we strengthen nontriviality to positivity. The distinction between
  nontriviality and positivity is analogous to the distinction between
  nonemptiness and inhabitedness.
  Moreover, we prove that locally small, nontrivial (directed or bounded)
  complete posets necessarily lack decidable equality.  We prove our results for
  a general class of posets, which includes \eg\ directed complete posets,
  bounded complete posets, sup-lattices and frames.
  Secondly, the fact that these nontrivial posets are necessarily large has the
  important consequence that Tarski's theorem (and similar results) cannot be
  applied in nontrivial instances. Furthermore, we explain that generalizations
  of Tarski's theorem that allow for large structures are provably false by
  showing that the ordinal of ordinals in a univalent universe has small suprema
  in the presence of set quotients.
  The latter also leads us to investigate the inter-definability and interaction
  of type universes of propositional truncations and set quotients, as well as a
  set replacement principle.
  Thirdly, we clarify, in our predicative setting, the relation between the
  traditional definition of sup-lattice that requires suprema for all subsets
  and our definition that asks for suprema of all small families.
\end{abstract}

\maketitle

\section{Introduction}\label{sec:introduction}

We investigate predicative aspects of constructive univalent
foundations. By~predicative and constructive, we respectively mean that we do
not assume Voevodsky's propositional resizing
axioms~\cite{Voevodsky2011,Voevodsky2015} or excluded middle and choice.
Most of our work is situated in our larger programme of developing domain theory
constructively and predicatively in univalent foundations. In~previous
work~\cite{deJongEscardo2021}, we showed how to give a constructive and
predicative account of many familiar constructions and notions in domain theory,
such as Scott's \(D_\infty\) model of untyped \(\lambda\)-calculus and the
theory of continuous dcpos. The present work complements this and other existing
work on predicative mathematics
(\eg~\cite{AczelRathjen2010,Sambin1987,CoquandEtAl2003}) by exploring what
\emph{cannot} be done predicatively, as
in~\cite{Curi2010a,Curi2010b,Curi2015,Curi2018,CuriRathjen2012}. We do so by
showing that certain statements crucially rely on resizing axioms in the sense
that they are equivalent to them. Such arguments are important in constructive
mathematics.  For example, the constructive failure of trichotomy on the real
numbers is shown~\cite{BridgesRichman1987} by reducing it to a nonconstructive
instance of excluded middle.

Our first main result is that nontrivial (directed or bounded) complete posets
are necessarily large. In~\cite{deJongEscardo2021} we observed that all our
examples of directed complete posets have large carriers. We show here that this
is no coincidence, but rather a necessity, in the sense that if such a
nontrivial poset is small, then weak propositional resizing holds. It is
possible to derive full propositional resizing if we strengthen nontriviality to
positivity in the sense of~\cite{Johnstone1984}. The distinction between
nontriviality and positivity is analogous to the distinction between
nonemptiness and inhabitedness.
We prove our results for a general class of posets, which includes directed
complete posets, bounded complete posets and sup-lattices, using a technical
notion of a \deltacomplete{\V} poset.
We also show that nontrivial locally small \deltacomplete{\V} posets necessarily
lack decidable equality. Specifically, we can derive weak excluded middle from
assuming the existence of a nontrivial locally small \deltacomplete{\V} poset
with decidable equality. Moreover, if we assume positivity instead of
nontriviality, then we can derive full excluded middle.

Secondly, the fact that these nontrivial posets are necessarily large has the
important consequence that Tarski's theorem (and similar results) cannot be
applied in nontrivial instances. Furthermore, we explain that generalizations of
Tarski's theorem that allow for large structures are provably
false. Specifically, we show that the ordinal of ordinals in a univalent
universe does not have a maximal element, but does have small suprema in the
presence of small set quotients.
The latter also leads us to investigate the inter-definability and interaction
of type universes of propositional truncations and set quotients, as well as a
set replacement principle.
Following a construction due to Voevodsky, we construct set quotients from
propositional truncations.
However, while Voevodsky assumed propositional resizing rules in his
construction, we show that, when propositional truncations are available,
resizing is not needed to prove the universal property of the set quotient, even
though the quotient will live in a higher type universe.

Finally, we clarify, in our predicative setting, the relation between the
traditional definition of sup-lattice that requires suprema for all subsets and
our definition that asks for suprema of all small families. This is important in
practice in order to obtain workable definitions of dcpo, sup-lattice, etc.\
in the context of predicative univalent mathematics.

Our foundational setup is the same as in~\cite{deJongEscardo2021}, meaning that
our work takes places in intensional Martin-L\"of Type Theory and adopts the
univalent point of view~\cite{HoTTBook}. This~means that we work with the
stratification of types as singletons, propositions (or subsingletons or truth
values), sets, {1-groupoids}, etc., and that we work with univalence. At
present, higher inductive types other than propositional truncation are not
needed. Often the only consequences of univalence needed here are functional and
propositional extensionality.  Two exceptions are
\cref{sec:smallness-and-univalence,sec:small-suprema-of-ordinals}. Full details
of our univalent type theory are given at the start
of~\cref{sec:foundations-and-small-types}.

\subsection{Reasons for studying predicativity}
We briefly describe some motivations for studying impredicativity in the form of
propositional resizing in univalent type theory.
The first reason for our interest is that, unlike the univalence axiom in
cubical type theory~\cite{CohenEtAl2018}, there is at present no known
computational interpretation of propositional resizing axioms.

Another reason for being interested in predicativity is the fact that
propositional resizing axioms fail in some models of univalent type theory. A
notable example of such a model is Uemura's cubical assembly
model~\cite{Uemura2019}. What is particularly striking about Uemura's model is
that it does support an impredicative universe \(\U\) in the sense that if \(X\)
is \emph{any} type and \(Y : X \to \U\), then \(\Pi_{x : X}Y(x)\) is in \(\U\)
again even if \(X\) isn't, but that propositional resizing fails for this
universe.
On the model-theoretic side, we also highlight Swan's (unpublished)
results~\cite{Swan2019a,Swan2019b} that show that propositional resizing axioms
fail in certain presheaf (cubical) models of type theory. Interestingly, Swan's
argument works by showing that the models violate certain collection principles
if we assume Brouwerian continuity principles in the metatheory.

By contrast, we should mention that propositional resizing is validated in many
models when a classical metatheory is assumed. For example, this is true for any
type-theoretic model topos~\cite[Proposition~11.3]{Shulman2019}. In particular,
Voevodsky's simplical sets model~\cite{KapulkinLumsdaine2021} validates excluded
middle and hence propositional resizing.
We note, however, that in other models it is possible for propositional resizing
to hold and excluded middle to fail, as shown
by~\cite[Remark~11.24]{Shulman2015}.

Another interesting aspect of impredicativity is that it is expected, by analogy
to predicative and impredicative set theories, that adding resizing axioms
significantly increases the proof-theoretic strength of univalent type
theory~\cite[Remark~1.2]{Shulman2019}.

This paper concerns resizing \emph{axioms}, meaning we ask a given type to be
\emph{equivalent} to one in some fixed universe \(\U\) of ``small'' types.
Voevodsky~\cite{Voevodsky2011} originally introduced resizing \emph{rules} which
add judgements and hence modify the syntax of the type theory to make the given
type inhabit \(\U\), rather than only asking for an equivalent copy
in~\(\U\). It is not known whether Voevodsky's resizing rules are consistent
with univalent type theory in the sense that no-one has constructed a model of
univalent type theory extended with such resizing rules, or proved a
contradiction in the system.
It is also an open problem~\cite[Section~10]{CohenEtAl2018} whether we have
normalization for cubical type theory extended with resizing rules.
In fact, as far as we know, this is an open problem for plain Martin-L\"of Type
Theory as well.

Lastly, one may have philosophical reservations regarding impredicativity. For
example, some constructivists may accept predicative set theories like Aczel's
CZF and Myhill's CST, but not Friedman's impredicative set theory IZF.
Or, paraphrasing Shulman's narrative~\cite{Shulman2011}, one can ask why
propositions (or (\(-1\))-types) should be treated differently, \ie\ given that
we have to take size seriously for \(n\)-types for \(n > -1\), why not do the
same for (\(-1\))-types?

\subsection{Related work}
Curi investigated the limits of predicative mathematics
in CZF~\cite{AczelRathjen2010} in a series of
papers~\cite{Curi2010a,Curi2010b,Curi2015,Curi2018,CuriRathjen2012}.
In particular, Curi shows (see~\cite[Theorem~4.4 and
  Corollary~4.11]{Curi2010a}, \cite[Lemma~1.1]{Curi2010b} and
  \cite[Theorem~2.5]{Curi2015}) that CZF cannot prove that various nontrivial
posets, including sup-lattices, dcpos and frames, are small. This result is
obtained by exploiting that CZF is consistent with the anti-classical
generalized uniformity principle
(GUP)~\cite[Theorem~4.3.5]{vandenBerg2006}.
Our related \cref{nontrivial-impredicativity,positive-impredicativity} is of a
different nature in two ways.
Firstly, our theorem is in the spirit of reverse constructive
mathematics~\cite{Ishihara2006}: Instead of showing that GUP implies that there
are no non-trivial small dcpos, we show that the existence of a non-trivial
small dcpo is \emph{equivalent} to weak propositional resizing, and that the
existence of a positive small dcpo is \emph{equivalent} to full propositional
resizing. Thus, if we wish to work with small dcpos, we are forced to assume
resizing axioms.
Secondly, we work in univalent foundations rather than CZF.  This may seem a
superficial difference, but a number of arguments in Curi's
papers~\cite{Curi2015,Curi2018} crucially rely on set-theoretical notions and
principles such as transitive set, set-induction, and the weak regular extension
axiom (wREA), which cannot even be formulated in the underlying type theory of
univalent foundations.
Moreover, although Curi claims that the arguments of~\cite{Curi2010a,Curi2010b}
can be adapted to some version of Martin-L\"of Type Theory, it is presently not
clear whether there is any model of univalent foundations which validates
GUP. However, one of the reviewers suggested that Uemura's cubical assemblies
model~\cite{Uemura2019} might validate it. In particular, the reviewer hinted
that \cite[Proposition~21]{Uemura2019} may be seen as a uniformity principle.

Finally, the construction of set quotients using propositional truncations is
due to Voevodsky and also appears in~\cite[Section~6.10]{HoTTBook} and
\cite[Section~3.4]{RijkeSpitters2015}. While Voevodsky assumed resizing rules
for his construction, we investigate the inter-definability of propositional
truncations and set quotients in the absence of propositional resizing axioms.

\subsection{Organization}
\emph{\cref{sec:foundations-and-small-types}}: Foundations and size matters,
including impredicativity, relation to excluded middle, univalence and closure
under embedded retracts.
\emph{\cref{sec:set-quotients-truncations-set-replacement}}: Inter-definability
of set quotients and propositional truncations, and equivalence of small set
quotients and set replacement.
\emph{\cref{sec:large-posets}}: Nontrivial and positive \deltacomplete{\V}
posets and reductions to impredicativity and excluded middle.
\emph{\cref{sec:maximal-and-fixed-points}}: Predicative unavailability of
Tarski's fixed point theorem and Pataraia's lemma, and suprema of ordinals.
\emph{\cref{sec:families-and-subsets}}: Comparison of completeness with respect
to families and with respect to subsets.
\emph{\cref{sec:conclusion}}: Conclusion and future work.


\section{Foundations and Small Types}\label{sec:foundations-and-small-types}
We work with a subset of the type theory described in~\cite{HoTTBook} and we
mostly adopt the terminological and notational conventions
of~\cite{HoTTBook}. We include \(+\)~(binary sum), \(\Pi\)~(dependent
products), \(\Sigma\)~(dependent sum), \(\Id\) (identity type), and inductive
types, including~\(\Zero\)~(empty type), \(\One\)~(type with exactly one element
\(\star : \One\)), \(\Nat\)~(natural numbers).
We assume a universe \(\U_0\) and two operations: for every universe \(\U\), a
successor universe \(\U^+\) with \(\U : \U^+\), and for every two universes
\(\U\) and \(\V\) another universe \(\U \sqcup \V\) such that for any
universe~\(\U\), we have \(\U_0 \sqcup \U \equiv \U\) and
\(\U \sqcup \U^+ \equiv \U^+\). Moreover, \((-)\sqcup(-)\) is idempotent,
commutative, associative, and \((-)^+\) distributes over \((-)\sqcup(-)\). We
write \(\U_1 \colonequiv \U_0^+\), \(\U_2 \colonequiv \U_1^+, \dots\) and so on.
If \(X : \U\) and \(Y : \V\), then \({X + Y} : \U \sqcup \V\) and if \(X : \U\)
and \(Y : X \to \V\), then the types \(\Sigma_{x : X} Y(x)\) and
\(\Pi_{x : X} Y(x)\) live in the universe \(\U \sqcup \V\); finally,
if~\(X : \U\) and \(x,y : X\), then \(\Id_{X}(x,y) : \U\). The type of natural
numbers \(\Nat\) is assumed to be in \(\U_0\) and we postulate that we have
copies \(\Zero_{\U}\) and \(\One_{\U}\) in every universe \(\U\).
This has the useful consequence that while we do not assume cumulativity of
universes, embeddings that lift types to higher universes are definable. For
example, the map \((-) \times \One_{\V}\) takes a type in any universe \(\U\) to
an equivalent type in the higher universe \(\U \sqcup \V\).
We assume function extensionality and propositional extensionality tacitly, and
univalence explicitly when needed. Finally, we use a single higher inductive
type: the propositional truncation of a type \(X\) is denoted by \(\squash*{X}\)
and we write \(\exists_{x : X}Y(x)\) for \(\squash*{\Sigma_{x : X}Y(x)}\).
Apart from \cref{sec:set-quotients-truncations-set-replacement}, we assume
throughout that every universe is closed under propositional truncations,
meaning that if \(X : \U\) then \(\squash{X} : \U\) as well.

\subsection{The Notion of a Small Type}
We introduce the fundamental notion of a type being \(\U\)-small with respect to
some type universe \(\U\), and specify the impredicativity axioms under
consideration~(\cref{sec:impred-and-em}). We also note the relation to excluded
middle~(\cref{sec:impred-and-em}) and
univalence~(\cref{sec:smallness-and-univalence}). Finally,
in~\cref{sec:small-types-and-retracts} we establish our main technical result on
small types, namely that being small is closed under retracts.

\begin{defi}[Smallness, {\cite[{\mkTTurl{UF.Size}}]{TypeTopology}}] 
    A type \(X\) in any universe is said to be \emph{\(\U\)-small} if it is
    equivalent to a type in the universe \(\U\). That is,
    \({X \issmall{\U}} \colonequiv \Sigma_{Y : \U} \pa*{Y \simeq X}\).
\end{defi}

\begin{defi}[Local smallness, \cite{Rijke2017}]
  A type \(X\) is said to be \emph{locally \(\U\)-small} if the type
  \((x = y)\) is \(\U\)-small for every \(x,y : X\).
\end{defi}

\begin{exas}\hfill
  \begin{enumerate}[(i)]
  \item Every \(\U\)-small type is locally \(\U\)-small.
  \item The type \(\Omega_{\U}\) of propositions in a universe \(\U\) lives in
    \(\U^+\), but is locally \(\U\)-small by propositional extensionality.
  \end{enumerate}
\end{exas}

\subsection{Impredicativity and Excluded Middle}
\label{sec:impred-and-em}
We consider various impredicativity axioms and their relation to (weak) excluded
middle. The definitions and propositions below may be found in
\cite[Section~3.36]{Escardo2019}, so proofs are omitted here.

\begin{defi}[Impredicativity axioms]
  \hfill
  \begin{enumerate}[(i)]
  \item By \emph{Propositional-\(\text{Resizing}_{\U,\V}\)} we mean the
    assertion that every proposition \(P\) in a universe \(\U\) is
    \(\V\)-small.
  \item We write \emph{\(\Omegaresizing{\U}{\V}\)} for the assertion that the
    type \(\Omega_{\U}\) is \(\V\)-small.
  \item The type of all \(\lnot\lnot\)-stable propositions in a universe \(\U\)
    is denoted by \(\Omeganotnot{\U}\), where a proposition \(P\) is
    \emph{\(\lnot\lnot\)-stable} if \(\lnot\lnot P\) implies \(P\).
    By \emph{\(\Omeganotnotresizing{\U}{\V}\)} we mean the assertion that the
    type \(\Omeganotnot{\U}\) is \(\V\)-small.
  \item For the particular case of a single universe, we write
    \(\Omegaresizingalt{\U}\) and \(\Omeganotnotresizingalt{\U}\) for the
    respective assertions that \(\Omega_{\U}\) is \(\U\)-small and
    \(\Omeganotnot{\U}\) is \(\U\)-small.
  \end{enumerate}
\end{defi}

\begin{prop}
  \hfill
  \begin{enumerate}[(i)]
  \item The principle \(\Omegaresizing{\U}{\V}\) implies
    \(\Propresizing{\U}{\V}\) for every two universes \(\U\) and \(\V\).
  \item The conjunction of \(\Propresizing{\U}{\V}\) and \(\Propresizing{\V}{\U}\)
  implies \(\Omegaresizing{\U}{\V^+}\) for every two universes \(\U\) and \(\V\).
  \end{enumerate}
\end{prop}
It is possible to define a weaker variation of propositional resizing for the
\(\lnot\lnot\)-stable propositions only (and derive similar connections), but we
don't need it in this paper.

\begin{defi}[(Weak) excluded middle]
  \hfill
  \begin{enumerate}[(i)]
  \item \emph{Excluded middle} in a universe \(\U\) asserts that for every
    proposition \(P\) in \(\U\) either \(P\)~or~\(\lnot P\) holds.
  \item \emph{Weak excluded middle} in a universe \(\U\) asserts that for every
    proposition \(P\) in \(\U\) either \(\lnot P\) or \(\lnot\lnot P\) holds.
  \end{enumerate}
\end{defi}
We note that weak excluded middle says precisely that \(\lnot\lnot\)-stable
propositions are decidable and is equivalent to de~Morgan's Law.

\begin{prop}
  Excluded middle implies impredicativity. Specifically,
  \begin{enumerate}[(i)]
  \item Excluded middle in \(\U\) implies \(\Omegaresizing{\U}{\U_0}\).
  \item Weak excluded middle in \(\U\) implies
    \(\Omeganotnotresizing{\U}{\U_0}\).
  \end{enumerate}
\end{prop}

\subsection{Smallness and Univalence}
\label{sec:smallness-and-univalence}
With univalence we can prove that the statements \(\Propresizing{\U}{\V}\) and
\(\Omegaresizing{\U}{\V}\) are subsingletons. More generally, univalence allows
us to prove that the statement that \(X\) is \(\V\)-small is a proposition,
which is needed at the end of~\cref{sec:small-completeness-with-resizing}.

\begin{propC}[{\cite[\texttt{has-size-is-subsingleton}]{Escardo2019}}]%
  \label{is-small-is-prop}
  If \(\V\) and \(\U \sqcup \V\) are univalent universes, then
  \(X \issmall{\V}\) is a proposition for every \(X : \U\).
\end{propC}
The converse also holds in the following form.
\begin{prop}\label{is-small-univalence}
  The type \(X \issmall{\U}\) is a proposition for
  every \(X : \U\) if and only if the universe \(\U\) is univalent.
\end{prop}
\begin{proof}
  Since \({X \issmall{\U}} \colonequiv \Sigma_{Y : \U}(Y\simeq X)\), this
  follows from~\cite[Section~3.14]{Escardo2019}. \qedhere
\end{proof}

\subsection{Small Types and Retracts}
\label{sec:small-types-and-retracts}
We show our main technical result on small types here, namely that being small
is closed under retracts.

\begin{defi}[Sections and retractions]
  A \emph{section} is a map \(s : X \to Y\) together with a left inverse
  \(r : Y \to X\), \ie\ the maps satisfy \(r \circ s \sim \id\).  We call \(r\)
  the \emph{retraction} and say that \(X\) is a \emph{retract} of \(Y\).
\end{defi}

We extend the notion of a small type to functions as follows.

\begin{defi}[Smallness (for maps), {\cite[{\mkTTurllong{UF.Size}{_is-small-map}}]{TypeTopology}}]
  A map \(f : X \to Y\) is said be \(\V\)-small if every fibre is \(\V\)-small.
\end{defi}

\begin{lemC}[{\cite[{\mkTTurllong{UF.Size}{_is-small-map}}]{TypeTopology}}]%
  \label{small-maps-lemmas}\hfill
  \begin{enumerate}[(i)]
  \item\label{small-type-in-terms-of-small-map} A type \(X\) is \(\V\)-small if
    and only if the unique map \(X \to \One_{\U_0}\) is \(\V\)-small.
  \item\label{small-domain-iff-small-map} If \(Y\) is \(\V\)-small, then a map
    \(f : X \to Y\) is \(\V\)-small if and only if \(X\) is.
  \end{enumerate}
\end{lemC}
\begin{proof}
  \ref{small-type-in-terms-of-small-map} Writing \(!_X\) for the map
  \(X \to \One_{\U_0}\) we have \(\fib_{!_X}(\star) \simeq X\).
  \ref{small-domain-iff-small-map} If \(X\) and \(Y\) are both
  \(\V\)\nobreakdash-small, witnessed respectively by \(\varphi : X' \simeq X\)
  and \(\psi : Y' \simeq Y\), then \(\fib_f(y)\) is \(\V\)\nobreakdash-small for
  every \(y : Y\), because
  \(\fib_f(y) \equiv \Sigma_{x : X}\pa*{f(x) = y} \simeq \Sigma_{x' :
    X'}\pa*{\psi^{-1}\pa*{f(\varphi(x'))} = \psi^{-1}(y)}\).
  Conversely, if \(f\) and \(Y\) are \(\V\)-small, then so is \(X\), because
  \cite[Lemma~4.8.2]{HoTTBook} tells us that
  \(X \simeq \Sigma_{y : Y}\fib_f(y)\).
\end{proof}

\begin{thm}\label{is-small-retract}
  Every section into a \(\V\)-small type is \(\V\)-small. In particular, its
  domain is \(\V\)-small. Hence, the \(\V\)-small types are closed under
  retracts.
\end{thm}
\begin{proof}
  We show that the domain is \(\V\)-small from which it follows that the section
  is \(\V\)-small by~\cref{small-maps-lemmas}\ref{small-domain-iff-small-map}.
  So suppose we have a section \(s : X \to Y\) with retraction \(r : Y \to X\)
  and that \(Y\) is \(\V\)-small.
  By~\cite[Lemma~3.6]{Shulman2016}, the endomap \(f \colonequiv r \circ s\) on
  \(Y\) is a quasi-idempotent~\cite[Definition~3.5]{Shulman2016}. Hence,
  \cite[Theorem~5.3]{Shulman2016} tells us that \(f\) can be split as
  \(Y \xrightarrow{r'} A \xrightarrow{s'} Y\) for some maps \(s'\) and \(r'\).
  Now \(X\) and \(A\) are equivalent as witnessed by the maps
  \(x \mapsto r'(s(x))\) and \(a \mapsto r(s'(a))\).
  Finally, we recall from the proof of~\cite[Theorem~5.3]{Shulman2016} that
  \(A \colonequiv \Sigma_{\sigma : \Nat \to Y}\Pi_{n : \Nat}\pa*{f(\sigma_{n+1})
    = \sigma_n}\) which is \(\V\)-small because \(Y\) is assumed to be.
\end{proof}

\begin{rem}\label{small-retract-improvement}
  In~\cite{FSCDversion} we had a weaker version of~\cref{is-small-retract} where
  we included the additional assumption that the section was an embedding. (Note
  that if every section is an embedding, then every type is a
  set~\cite[Remark~3.11(2)]{Shulman2016}, but that all sections into \emph{sets}
  are embeddings~\cite[\texttt{lc-maps-into-sets-are-embeddings}]{Escardo2019}.)
  We are grateful to the anonymous reviewer who proposed the above
  strengthening.
\end{rem}

\section{Set Quotients, Propositional Truncations and Set Replacement}
\label{sec:set-quotients-truncations-set-replacement}
We investigate the inter-definability and interaction of type universe levels of
propositional truncations and set quotients in the absence of propositional
resizing axioms. In particular, we will see that it is not so important if the
set quotient or propositional truncation lives in a higher universe. What is
paramount instead is whether the universal property applies to types in
arbitrary universes.
However, in some cases, like in~\cref{sec:small-suprema-of-ordinals}, it is
relevant whether set quotients are small and we show this to be equivalent to a
set replacement principle in~\cref{sec:set-replacement}.

We start by recalling (the universal property of) the propositional truncation,
which, borrowing terminology from category theory, we could also call the
\emph{subsingleton reflection} or \emph{propositional reflection}.

\begin{defi}[Propositional truncation, \(\squash{-}\)]\label{def:prop-trunc}
  A \emph{propositional truncation} of a type~\(X\), if it exists, is a
  proposition \(\squash{X}\) with a map \(\tosquash{-} : X \to \squash{X}\) such
  that every function \(f : X \to P\) to any proposition factors through
  \(\tosquash{-}\).
  \begin{center}
  \begin{tikzcd}
    X \ar[dr,"f"'] \ar[rr, "\tosquash{-}"] & & \squash{X} \ar[dl,dashed,"\bar{f}"] \\
    & P
  \end{tikzcd}
  \end{center}
\end{defi}

Some sources, \eg~\cite{HoTTBook}, also demand that the diagram above commutes
\emph{definitionally}: for every \(x : X\), we have
\(f(x) \equiv \bar{f}(\tosquash{x})\). Having definitional equalities has some
interesting consequences, such as being able to prove function
extensionality~\cite[Section~8]{KrausEtAl2017}.
We do not require definitional equalities, but notice that
we do have \(f(x) = \bar{f}(\tosquash{x})\) (up to an identification) for every
\(x : X\), as \(P\) is a subsingleton.
In particular it follows using function extensionality that \(\bar{f}\) is the
unique factorization.

Notice that if a propositional truncation exists, then it is unique up to unique
equivalence.

\begin{rem}\label{prop-trunc-universes}
Some remarks regarding universes are in order:
\begin{enumerate}[(i)]
\item In~\cref{def:prop-trunc}, the subsingleton \(P\) may live in an
  \emph{arbitrary} universe, regardless of the universe in which \(X\) sits. The
  importance of this will be revisited throughout this section and
  in~\cref{Voevodsky-prop-trunc-func} in particular.
\item In~\cref{def:prop-trunc}, we haven't specified in what universe
  \(\squash{X}\) should be. When adding propositional truncations as higher
  inductive types, one typically assumes that \(\squash{X} : \U\) if \(X : \U\),
  and indeed this is what we do in most of this paper. In this section,
  however, we will be more general and instead assume that
  \(\squash{X} : F(\U)\) where \(F\) is a (meta)function on universes, so that
  the above case is obtained by taking \(F\) to be the identity. We will also
  consider \(F(\U) = \U_1 \sqcup \U\) in the final subsection.
\end{enumerate}
\end{rem}

While in general propositional truncations may fail to exist in intensional
Martin-L\"of Type Theory, it is possible to construct a propositional truncation
of some types in specific cases~\cite[Section~3.1]{EscardoXu2015}. A particular
example~\cite[Corollary~4.4]{KrausEtAl2017} is for a type \(X\) with a weakly
constant (viz.\ any of its values are equal) endofunction \(f\): the
propositional truncation of \(X\) can be constructed as
\(\Sigma_{x : X}(x=f(x))\), the type of fixed points of \(f\).

We review an approach by Voevodsky, who used resizing rules, to
constructing propositional truncations in general in the next section.

\subsection{Propositional Truncations and Propositional Resizing}
Voevodsky~\cite{Voevodsky2011} introduced propositional resizing rules in order
to construct propositional
truncations~\cite[Section~2.4]{PelayoVoevodskyWarren2015}. Here we review
Voevodsky's construction, paying special attention to the universes involved.

\emph{NB.\ We do not assume the availability of propositional truncations in
this section.}

\begin{defi}[Voevodsky propositional truncation, \(\squashVV{X}\)]
  The \emph{Voevodsky propositional truncation} \(\squashVV{X}\) of a type
  \(X : \U\) is defined as
  \[
    \squashVV{X} \colonequiv \prod_{P : \U}
    \pa*{\issubsingleton(P) \to (X \to P) \to P}.
  \]
\end{defi}

Because of function extensionality, one can show that \(\squashVV{X}\) is indeed
a proposition for every type \(X\). Moreover, we have a map
\(\tosquashVV{-} : X \to \squashVV{X}\) given by
\(\tosquashVV{x} \colonequiv (P,i,f) \mapsto f(x)\).

Observe that \(\squashVV{X} : \U^+\), so using the notation
from~\cref{prop-trunc-universes}, we have \(F(\U) = \U^+\).
However, as we will argue for set quotients, it does not matter so much where
the truncated proposition lives; it is much more important that we can eliminate
into subsingletons in arbitrary universes, \ie\ that \(\squashVV{-}\) satisfies
the right universal property.
Given \(X : \U\) and a map \(f : X \to P\) to a proposition \(P : \U\) with
\(i : \issubsingleton(P)\), we have a map \(\squashVV{X} \to P\) given as
\(\Phi \mapsto \Phi(P,i,f)\).
However, if the proposition \(P\) lives in some other universe \(\V\), then we
seem to be completely stuck. To clarify this, we consider the example of
functoriality.

\begin{exa}\label{Voevodsky-prop-trunc-func}
  If we have a map \(f : X \to Y\) with \(X : \U\) and \(Y : \U\), then we get a
  lifting simply by precomposition, \ie\ we define
  \(\tosquashVV{f} : \squashVV{X} \to \squashVV{Y}\) by
  \(\tosquashVV{f}(\Phi) \colonequiv (P,i,g) \mapsto \Phi(P,i,g \circ f)\).
  But obviously, we also want functoriality for maps \(f : X \to Y\) with
  \(X : \U\) and \(Y : \V\), but this is impossible with the above definition of
  \(\tosquashVV{f}\), because for \(\squashVV{X}\) we are considering
  propositions in \(\U\), while for \(\squashVV{Y}\) we are considering
  propositions in \(\V\).

  In particular, even if the types \(X : \U\) and \(Y : \V\) are equivalent,
  then it does not seem possible to construct an equivalence between
  \(\squashVV{X}\) and \(\squashVV{Y}\). This issue also comes up if one tries
  to prove that the map \(\tosquashVV{-} : X \to \squashVV{X}\) is a
  surjection~\cite[Section~3.34.1]{Escardo2019}.
\end{exa}

\begin{propC}[{\cite[Theorem~3.8]{KrausEtAl2017}}]
  \label{prop-trunc-implies-resizing-of-VV-trunc}
  If our type theory has propositional truncations with \(\squash{X} : \U\)
  whenever \(X : \U\), then \(\squashVV{X}\) is \(\U\)-small.
\end{propC}
\begin{proof}
  We will show that \(\squash{X}\) and \(\squashVV{X}\) are logically equivalent
  (\ie\ we have maps in both directions), which suffices, because both types
  are subsingletons. We obtain a map \(\squash{X} \to \squashVV{X}\) by applying
  the universal property of \(\squash{X}\) to the map
  \(\tosquashVV{-} : X \to \squashVV{X}\). Observe that it is essential that the
  universal property allows for elimination into subsingletons in universes
  other than \(\U\), as \(\squashVV{X} : \U^+\). For the function in the other
  direction, simply note that \(\squash{X} : \U\), so that we can construct
  \(\squashVV{X} \to \squash{X}\) as
  \(\Phi \mapsto \Phi(\squash{X},i,\tosquash{-})\) where
  \(i\) witnesses that \(\squash{X}\) is a subsingleton.
\end{proof}
Thus, as is folklore in the univalent foundations community, we can view
higher inductive types as specific resizing axioms. But note that the converse
to the above proposition does not appear to hold, because even if
\(\squashVV{X}\) is \(\U\)-small, then it still wouldn't have the appropriate
universal property. This is because the definition of \(\squashVV{X}\) is a
dependent product over propositions in \(\U\) only, which now includes
\(\squashVV{X}\), but still misses propositions in other universes.
In the presence of resizing axioms, we could obtain the full universal property,
because we would have (equivalent copies of) all propositions in a single
universe:

\begin{prop}[see \eg~{\cite[Section~36.5]{Escardo2019}}]
  If \(\Propresizing{\U}{\U_0}\) holds for every universe \(\U\), then the
  Voevodsky proposition truncation satisfies the full universal property with respect
  to all types in all universes.
\end{prop}

\subsection{Set Quotients from Propositional Truncations}
\label{set-quotients-from-propositional-truncations}

In this section we assume to have propositional truncations with
\(\squash{X} : F(\U)\) when \(X : \U\) for some (meta)function~\(F\) on
universes. We will be mainly interested in \(F(\U) = \U\) and
\(F(\U) = \U_1 \sqcup \U\) for the reasons explained below.
We prove that we can construct set quotients using propositional
truncations. The construction is due to Voevodsky and also appears
in~\cite[Section~6.10]{HoTTBook} and \cite[Section~3.4]{RijkeSpitters2015}. %
However, while Voevodsky assumed propositional resizing rules in his
construction, the point of this section is to show that resizing is not needed
to prove the universal property of the set quotient, provided propositional
truncations are available.
Our proof follows our earlier Agda development~\cite{TypeTopologyQuotientLarge}
(see also~\cite[Section~3.37]{Escardo2019}) and is fully
formalized~\cite{TypeTopologyQuotientF}.

\subsubsection{Images and Surjections}
\label{sec:image-and-surjection}
It will be convenient to first state and prove two lemmas on images and
surjections.

\begin{defi}[Image, \(\image(f)\), surjection, corestriction]
  \hfill
  \begin{enumerate}
  \item The \emph{image} of a function \(f : X \to Y\) is defined as
    \( \image(f) \colonequiv \Sigma_{y : Y}\exists_{x : X}(f(x) = y) \).
  \item A function \(f : X \to Y\) is a \emph{surjection} if for every
    \(y : Y\), there exists some \(x : X\) such that \(f(x) = y\).
  \item The \emph{corestriction} of a function \(f : X \to Y\) is the function
    \(f : X \to \image(f)\) given by
    \(x \mapsto \pa*{f(x),\tosquash*{x,\refl}}\).
  \end{enumerate}
\end{defi}

\begin{rem}\label{image-universes}
  Note that if \(X : \U\) and \(Y : \V\) and \(f : X \to Y\), then
  \(\image(f) : \V \sqcup F(\U \sqcup \V)\), because
  \(\Sigma_{x : X}(f(x) = y) : \U \sqcup \V\) and \(\squash{-}\) takes types in
  \(\W\) to subsingletons in \(F(\W)\).
  In case \(F\) is the identity, then we obtain the simpler
  \(\image(f) : \U \sqcup \V\).
\end{rem}

\begin{lem}\label{corestriction-is-surjective}
  Every corestriction is surjective.
\end{lem}
\begin{proof}
  By definition of the corestriction.
\end{proof}

\begin{lem}[Image induction,
  {\cite[\mkTTurllong{UF.ImageAndSurjection}{image-induction}]{TypeTopology}}]%
  \label{image-induction}
  For a surjective map \(f : X \to Y\), the following induction principle
  holds: for every prop-valued \(P : Y \to \W\), with \(\W\) an
  \emph{arbitrary} universe, if \(P(f(x))\) holds for every \(x : X\), then
  \(P(y)\) holds for every \(y : Y\).

  In the other direction, for any map \(f : X \to Y\), if the above induction
  principle holds for the specific family
  \({P(y) \colonequiv \exists_{x : X}(f(x) = y)}\), then \(f\) is a surjection.
\end{lem}
\begin{proof}
  Suppose that \(f : X \to Y\) is a surjection, let \(P : Y \to \W\) be
  subsingleton-valued and assume that \(P(f(x))\) holds for every \(x : X\).
  Now let \(y : Y\) be arbitrary. We are to prove that \(P(y)\) holds. Since
  \(f\) is a surjection, we have \(\exists_{x : X}(f(x) = y)\). But \(P(y)\) is
  a subsingleton, so we may assume that we have a specific \(x : X\) with
  \(f(x) = y\). But then \(P(y)\) must hold, because \(P(f(x))\) does by
  assumption.

  For the other direction, notice that if
  \(P(y) \colonequiv \exists_{x : X}(f(x) = y)\), then \(P(f(x))\) clearly holds
  for every \(x : X\). So by assuming that the induction principle applies, we
  get that \(P(y)\) holds for every \(y : Y\), which says exactly that \(f\) is
  a surjection.
\end{proof}

\subsubsection{Set Quotients}\label{sec:set-quotients}
We now construct set quotients using images and specialize image induction to
the set quotient.
\begin{defi}[Equivalence relation]
  An \emph{equivalence relation} on a type \(X\) is a binary type family
  \({\approx} : X \to X \to \V\) such that it is
  \begin{enumerate}[(i)]
  \item subsingleton-valued, \ie\ \(x \approx y\) is a subsingleton for
    every \(x,y : X\);
  \item reflexive, \ie\ \(x \approx x\) for every \(x : X\);
  \item symmetric, \ie\ \(x \approx y\) implies \(y \approx x\) for
    every \(x,y : X\);
  \item transitive, \ie\ the conjunction of \(x \approx y\) and \(y \approx z\)
    implies \(x \approx z\) for every \(x,y,z : X\).
  \end{enumerate}
\end{defi}

\begin{defi}[Set quotient, \(X/{\approx}\)]
  We define the \emph{set quotient} of \(X\) by \(\approx\) to be the type
  \(X/{\approx} \colonequiv \image (e_\approx)\) where
  \begin{align*}
    e_\approx : X &\to (X \to \Omega_{\V}) \\
    x &\mapsto \pa*{y \mapsto (x \approx y,p(x,y))}
  \end{align*}
  and \(p\) is the witness that \(\approx\) is subsingleton-valued.
\end{defi}

Of course, we should prove that \(X/{\approx}\) really is the quotient of \(X\)
by \(\approx\) by proving a suitable universal property. The following
definition and lemmas indeed build up to this. For the remainder of this
section, we will fix a type \(X : \U\) with an equivalence relation
\({\approx} : X \to X \to \V\).

\begin{rem}\label{quotient-universes}
  By~\cref{image-universes}, and because \(\Omega_{\V} : \V^+\), we have
  \(X/{\approx} : \T \sqcup F(\T)\) with \(\T \colonequiv \V^+ \sqcup \U\).
  In the particular case that \(F\) is the identity, we obtain the simpler
  \(X/{\approx} : \V^+ \sqcup \U\).
\end{rem}

\begin{lem}
  The quotient \(X/{\approx}\) is a set.
\end{lem}
\begin{proof}
  Observe that %
  \(\pa*{X/{\approx}} \equiv \image(e_\approx)\) is a subtype of
  \(X \to \Omega_{\V}\) (as
  \({\fst : X/{\approx} \to (X \to \Omega_\V)}\) is an embedding), that
  \(X \to \Omega_{\V}\) is a set (by function extensionality) and that subtypes
  of sets are sets.
\end{proof}

\begin{defi}[\(\eta\)]
  The map \(\eta : X \to X/{\approx}\) is defined to be the corestriction of
  \(e_{\approx}\).
\end{defi}

Although, in general, the type \(X/{\approx}\) lives in another universe than
\(X\) (see~\cref{quotient-universes}), we can still prove the following
induction principle for (subsingleton-valued) families into \emph{arbitrary}
universes.

\begin{lem}[Set quotient induction]\label{set-quotient-induction}
  For every subsingleton-valued \(P : X/{\approx} \to \W\), with \(\W\)
  \emph{any} universe, if \(P(\eta(x))\) holds for every \(x : X\), then
  \(P(x')\) holds for every \(x' : X/{\approx}\).
\end{lem}
\begin{proof}
  The map \(\eta\) is surjective by~\cref{corestriction-is-surjective}, so
  that~\cref{image-induction} yields the desired result.
\end{proof}

\begin{defi}[Respect equivalence relation]
  A map \(f : X \to A\) \emph{respects the equivalence relation} \(\approx\) if
  \(x \approx y\) implies \(f(x) = f(y)\) for every \(x,y : X\).
\end{defi}

Observe that respecting an equivalence relation is property rather than data,
when the codomain \(A\) of the map \(f : X \to A\) is a set.

\begin{lem}\label{eta-respects-and-effective}
  The map \(\eta : X \to X/{\approx}\) respects the equivalence relation
  \({\approx}\) and the set quotient is \emph{effective}, \ie\ for every
  \(x,y : X\), we have \(x \approx y\) if and only if \(\eta(x) = \eta(y)\).
\end{lem}
\begin{proof}
  By definition of the image and function extensionality, we have for every
  \(x,y : X\) that \(\eta(x) = \eta(y)\) holds if and only if
  \begin{equation}
    \forall_{z : X}\pa*{x \approx z \iff y \approx z}
    \tag{\(\ast\)}\label{eta-equality}
  \end{equation}
  holds. If~\eqref{eta-equality} holds, then so does \(x \approx y\) by
  reflexivity and symmetry of the equivalence relation. Conversely, if
  \(x \approx y\) and \(z : X\) is such that \(x \approx z\), then
  \(y \approx z\) by symmetry and transitivity; and similarly if \(z : X\) is
  such that \(y \approx z\). Hence, \eqref{eta-equality} holds if and only if
  \(x \approx y\) holds.  Thus, \(\eta(x) = \eta(y)\) if and only if
  \(x \approx y\), as desired.
\end{proof}

The universal property of the set quotient states that the map
\(\eta : X \to X/{\approx}\) is the universal function to a set preserving the
equivalence relation. We can prove it using only~\cref{set-quotient-induction}
and~\cref{eta-respects-and-effective}, without the need to inspect the definition
of the quotient.

\begin{thm}[Universal property of the set quotient]
  \label{set-quotient-universal-property}
  For every \emph{set} \(A : \W\) in \emph{any} universe~\(\W\) and function
  \(f : X \to A\) respecting the equivalence relation, there is a unique
  function \(\bar{f} : X/{\approx} \to A\) such that the diagram
  \begin{center}
    \begin{tikzcd}
      X \ar["f"', dr] \ar["\eta", rr] & & X/{\approx} \ar["\bar{f}",dl,dashed] \\
      & A
    \end{tikzcd}
  \end{center}
  commutes.
\end{thm}
\begin{proof}[Proof~\cite{TypeTopologyQuotientF}]
  Let \(A : \W\) be a set and \(f : X \to A\) respect the equivalence relation.
  The~following auxiliary type family over \(X/{\approx}\) will be at the heart
  of our proof:
  \[
    B(x') \colonequiv
    \Sigma_{a : A}\exists_{x : X}\pa{(\eta(x) = x') \times (f(x) = a)}.
  \]
  \begin{clmnn}
    The type \(B(x')\) is a subsingleton for every \(x' : X/{\approx}\).
  \end{clmnn}
  \begin{proof}[Proof of claim]
    By function extensionality, the type expressing that \(B(x')\) is a
    subsingleton for every \(x' : X/{\approx}\) is itself a subsingleton. So by
    set quotient induction, it suffices to prove that \(B(\eta(x))\) is a
    subsingleton for every \(x : X\). So assume that we have
    \((a,p) , (b,q) : B(\eta(x))\). It suffices to show that \(a = b\). The
    elements \(p\) and \(q\) witness
    \[
      \exists_{x_1 : X}\pa{(\eta(x_1) = \eta(x)) \times (f(x_1) = a)}
    \]
    and
    \[
      \exists_{x_2 : X}\pa{(\eta(x_2) = \eta(x)) \times (f(x_2) = b)},
    \]
    respectively. By~\cref{eta-respects-and-effective} and the fact that \(f\)
    respects the equivalence relation, we obtain \(f(x) = a\) and \(f(x) = b\)
    and hence the desired \(a = b\).
  \end{proof}
  Next, we define \(k : \Pi_{x : X} B(\eta(x))\) by
  \(k(x) = \pa*{f(x),\tosquash*{x,\refl,\refl}}\). By set quotient induction and
  the claim, the function \(k\) induces a dependent map
  \(\bar{k} : \Pi_{\pa*{x' : X/{\approx}}} B(x')\).

  We then define the (nondependent) function \(\bar{f} : X/{\approx} \to A\) as
  \({\fst} \circ {\bar{k}}\). We proceed by showing that
  \(\bar{f} \circ \eta = f\). By function extensionality, it suffices to prove
  that \(\bar{f}(\eta(x)) = f(x)\) for every \(x : X\). But notice that:
  \begin{align*}
    \bar{f}(\eta(x)) &\equiv \fst(\bar{k}(\eta(x))) \\
                     &= \fst(k(x)) &\text{(since \(\bar{k}(\eta(x)) = k(x)\)
                                           because of the claim)} \\
    &\equiv f(x).
  \end{align*}

  Finally, we wish to show that \(\bar{f}\) is the unique such function, so
  suppose that \(g : X/{\approx} \to A\) is another function such that
  \(g \circ \eta = f\). By function extensionality, it suffices to prove that
  \(g(x') = \bar{f}(x')\) for every \(x' : X/{\approx}\), which is a
  subsingleton as \(A\) is a set. Hence, set quotient induction tells us that it
  is enough to show that \(g(\eta(x)) = \bar{f}(\eta(x))\) for every \(x : X\),
  but this holds as both sides of the equation are equal to \(f(x)\).
\end{proof}

\begin{rem}[\cf\ {\cite[Section~3.21]{Escardo2019}}]
  In univalent foundations, some attention is needed in phrasing unique
  existence, so we pause to discuss the phrasing
  of~\cref{set-quotient-universal-property} here.
  Typically, if we wish to express unique existence of an element \(x : X\)
  satisfying \(P(x)\) for some type family \(P : \U \to \V\), then we should
  phrase it as \(\issingleton(\Sigma_{x : X}P(x))\), where
  \(\issingleton(Y) \colonequiv Y \times \issubsingleton(Y)\). That is, we
  require that there is a unique \emph{pair} \((x,p) : \Sigma_{x : X}P(x)\).
  This becomes important when the type family \(P\) is not
  subsingleton-valued. However, if \(P\) is subsingleton-valued, then it is
  equivalent to the traditional formulation of unique existence: \ie\ that
  there is an \(x : X\) with \(P(x)\) such that every \(y : X\) with \(P(y)\) is
  equal to \(x\).
  This happens to be the situation in~\cref{set-quotient-universal-property},
  because of function extensionality and the fact that \(A\) is a set.
\end{rem}

We stress that although the set quotient increases universe levels,
see~\cref{quotient-universes}, it does satisfy the appropriate universal
property, so that resizing is not needed.

Having small set quotients is closely related to propositional resizing, as we
show now.

\begin{prop}
  Suppose that \(\squash{-}\) does not increase universe levels, \ie\ in the
  notation of~\cref{prop-trunc-universes}, the function \(F\) is the identity.
  \begin{enumerate}[(i)]
  \item\label{quotient-resizing-1} %
    If \(\Omegaresizing{\V}{\U}\) holds for universes \(\U\) and \(\V\), then
    the set quotient \(X/{\approx}\) is \(\U\)-small for any type \(X : \U\) and
    any \(\V\)-valued equivalence relation.
  \item\label{quotient-resizing-2} %
    Conversely, if the set quotient \(\Two/{\approx}\) is \(\U\)-small for
    every \(\V\)-valued equivalence relation on \(\Two\), then
    \(\Propresizing{\V}{\U}\) holds.
  \end{enumerate}
\end{prop}
\begin{proof}\hfill%
  \begin{enumerate}[(i)]
  \item If we have \(\Omegaresizing{\V}{\U}\), then \(\Omega_{\V}\) is
    \(\U\)-small, so that \(X/{\approx} \equiv \image(e_{\approx})\) is
    \(\U\)-small too when \(X : \U\) and \({\approx}\) is \(\V\)-valued.
  \item Let \(P : \V\) be any proposition and consider the \(\V\)-valued
    equivalence relation \( x \approx_P y \colonequiv (x = y) \vee P \) on
    \(\Two\). Notice that
    \[
      \pa*{\Two/{\approx_P}} \text{ is a subsingleton} \iff P \text{ holds},
    \]
    so if \(\Two/{\approx_P}\) is \(\U\)-small, then so is the type
    \(\issubsingleton\pa*{\Two/{\approx_P}}\) and therefore \(P\).
    \qedhere
  \end{enumerate}
\end{proof}

\subsection{Propositional Truncations from Set Quotients}
\label{sec:propositional-truncations-from-set-quotients}
The converse, constructing propositional truncations from set quotients, is more
straightforward, although we must pay some attention to the universes involved
in order to get an exact match.

\begin{defi}[Existence of set quotients]\label{existence-of-set-quotients}
  We say that \emph{set quotients exist} if for every type \(X\) and equivalence
  relation \({\approx}\) on \(X\), we have a set \(X/{\approx}\) with a map
  \(\eta : X \to X/{\approx}\) that respects the equivalence relation such that
  the universal property set out in~\cref{set-quotient-universal-property} is
  satisfied.
\end{defi}

\begin{thm}
  Any set quotient satisfies the induction principle
  of~\cref{set-quotient-induction}, \ie\ the induction principle is implied by
  the universal property of the set quotient.
\end{thm}
\begin{proof}[Proof~{\cite{TypeTopologyQuotient}}]
  Suppose that \(P : X/{\approx} \to \W\) is a proposition-valued type-family
  over the set quotient \(X/{\approx}\) and that we have
  \(\rho : \Pi_{x : X}P(\eta(x))\). We write
  \(S \colonequiv \Sigma_{x' : X/{\approx}}\,P(x')\) and define the map
  \(f : X \to S\) by \(f(x) \colonequiv \pa*{\eta(x) , \rho(x)}\). Note that
  \(f\) respects the equivalence relation since \(\eta\)~does and \(P\) is
  proposition-valued.
  Moreover, \(S\) is a set, because subtypes of sets are sets and the quotient
  \(X/{\approx}\) is a set by assumption.
  Hence, by the universal property, \(f\) induces a map
  \(\bar{f} : X/{\approx} \to S\) such that \(\bar{f} \circ \eta = f\).
  We claim that \(\bar{f}\) is a section of \(\fst : S \to X/{\approx}\).
  Note that this would finish the proof, because if we have
  \(e : \Pi_{x' : X/{\approx}}\,\fst\pa*{\bar{f}(x')} = x'\), then we obtain
  \(P(x')\) for every \(x'\) by transporting \(\snd\pa*{\bar{f}(x')}\) along
  \(e(x')\).
  But \(\bar{f}\) must be a section of~\(\fst\), because we can take both
  \({\fst} \circ {\bar{f}}\) and \(\id\) for the dashed map in the commutative
  diagram
  \begin{center}
    \begin{tikzcd}
      X \ar["\eta"', dr] \ar["\eta", rr] & & X/{\approx} \ar[dl,dashed] \\
      & X/{\approx}
    \end{tikzcd}
  \end{center}
  since \({\fst} \circ {\bar{f}} \circ {\eta} = {\fst} \circ {f} = \eta\), so
  \({\fst} \circ {\bar{f}}\) and \(\id\) must be equal by the universal property
  of the set quotient.
\end{proof}

\begin{thm}\label{set-quotients-give-propositional-truncations}
  If set quotients exist, then every type has a propositional truncation.
\end{thm}
\begin{proof}[Proof~{\cite{TypeTopologyQuotient}}]
  Let \(X : \U\) be any type and consider the \(\U_0\)-valued equivalence
  relation \(x \approx_{\One} y \colonequiv \One\).
  To see that \(X/{\approx_{\One}}\) is a subsingleton, note that by set
  quotient induction it suffices to prove \(\eta(x) = \eta(y)\) for every
  \(x,y : X\).
  But \(x \approx_{\One} y\) for every \(x,y : X\), and \(\eta\) respects the
  equivalence relation, so this is indeed the case.
  Now if \(P : \V\) is any subsingleton and \(f : X \to P\) is any map, then
  \(f\) respects the equivalence relation \({\approx_{\One}}\) on~\(X\), simply
  because \(P\) is a subsingleton. Thus, by the universal property of the
  quotient, we obtain the desired map \(\bar{f} : X/{\approx_{\One}} \to P\) and
  hence, \(X/{\approx_{\One}}\) has the universal property of the propositional
  truncation.
\end{proof}

\begin{rem}
  Because the set quotients constructed using the propositional truncation live
  in higher universes, we embark on a careful comparison of universes. Suppose
  that propositional truncations of types \(X : \U\) exist and that
  \(\squash{X} : F(\U)\). Then by~\cref{quotient-universes}, the set quotient
  \(X/{\approx_{\One}}\) in the proof above lives in the universe
  \((\U_1 \sqcup \U) \sqcup F(\U_1 \sqcup \U)\).

  In particular, if \(F\) is the identity and the propositional truncation of
  \(X : \U\) lives in \(\U\), then the quotient \(X/{\approx_{\One}}\) lives in
  \(\U_1 \sqcup \U\), which simplifies to \(\U\) whenever \(\U\) is at least
  \(\U_1\). In other words, the universes of \(\squash{X}\) and
  \(X/{\approx_{\One}}\) match up for types \(X\) in every universe,
  \emph{except} the first universe \(\U_0\).

  If we always wish to have \(X/{\approx_{\One}}\) in the same universe as
  \(\squash{X}\), then we can achieve this by assuming
  \(F(\V) \colonequiv \U_1 \sqcup \V\), which says that the propositional
  truncations stay in the same universe, \emph{except} when the type is in the
  first universe \(\U_0\) in which case the truncation will be in the second
  universe \(\U_1\).
\end{rem}

\begin{thm}
  All set quotients are effective, \ie~\(\eta(x) = \eta(y)\) implies
  \(x \approx y\).
\end{thm}
\begin{proof}
  If we have set quotients, then we have propositional truncations
  by~\cref{set-quotients-give-propositional-truncations} which we can use to
  construct effective set quotients
  following~\cref{set-quotients-from-propositional-truncations}.
  But any two set quotients of a type by an equivalence relation must be
  equivalent, so the original set quotients are effective too.
\end{proof}

\subsection{Set Replacement}\label{sec:set-replacement}
In this section, we return to our running assumption that universes are closed
under propositional truncations, \ie\ the function \(F\) above is assumed to be
the identity. We study the equivalence of a set replacement principle and the
existence of small set quotients. These principles will find application
in~\cref{sec:small-suprema-of-ordinals}.

\begin{defi}[Set replacement, {\cite[\mkTTurl{UF.Size}]{TypeTopology}}]
  The \emph{set replacement} principle asserts that the image of a map
  \(f : X \to Y\) is \(\U\sqcup\V\)-small if \(X\) is \(\U\)-small and \(Y\) is
  locally \(\V\)-small set.
\end{defi}

In particular, if \(\U\) and \(\V\) are the same, then the image is
\(\U\)-small.
The name ``set replacement'' is inspired by~\cite[Section~2.19]{BezemEtAl2022},
but is different in two ways: In \cite{BezemEtAl2022}, replacement is not
restricted to maps into sets, and the universe parameters \(\U\)~and~\(\V\) are
taken to be the same.
Rijke~\cite{Rijke2017} shows that the replacement of~\cite{BezemEtAl2022} is
provable in the presence of a univalent universe closed under pushouts.

We show that set replacement is logically equivalent to having small set
quotients, where the latter means that the quotient of a type \(X : \U\) by a
\(\V\)-valued equivalence relation lives in \(\U \sqcup \V\).

\begin{defi}[Existence of small set quotients]\label{existence-of-small-set-quotients}
  We say that \emph{small set quotients exist} if set quotients exists in the
  sense of~\cref{existence-of-set-quotients}, and moreover, the quotient
  \(X/{\approx}\) of a type \(X : \U\) by a \(\V\)-valued equivalence relation
  lives in \(\U \sqcup \V\).
\end{defi}

Note that we would get small set quotients if we added set quotients as a
primitive higher inductive type. Also, if one assumes \(\Omegaresizingalt{\V}\),
then the construction of set quotients in~\cref{sec:set-quotients} yields a
quotient \({X/{\approx}}\) in \({\U \sqcup \V}\) when \(X : \U\) and
\({\approx}\) is a \(\V\)-valued equivalence relation on \(X\).

\begin{thm}
  Set replacement is logically equivalent to the existence of small set
  quotients.
\end{thm}
\begin{proof}[Proof~{\cite{TypeTopologyQuotient,TypeTopologyQuotientReplacement}}]
  Suppose set replacement is true and that a type \(X : \U\) and a \(\V\)-valued
  equivalence relation~\({\approx}\) are given. Using the construction laid out
  in~\cref{sec:set-quotients}, we construct a set quotient \(X/{\approx}\) in
  \(\U \sqcup \V^+\) as the image of a map \(X \to (X \to \Omega_{\V})\). But by
  propositional extensionality \(\Omega_{\V}\) is locally \(\V\)-small and by
  function extensionality so is \(X \to \Omega_{\V}\). Hence, \(X/{\approx}\) is
  \((\U \sqcup \V)\)-small by set replacement, so \(X/{\approx}\) is equivalent to
  a type \(Y : \U\sqcup\V\).  It is then straightforward to show that \(Y\)
  satisfies the properties of the set quotient as well, finishing the proof of one
  implication.

  Conversely, let \(f : X \to Y\) be a map from a \(\U\)-small type to a locally
  \(\V\)-small set. Since \(X\) is \(\U\)-small, we have \(X' : \U\) such that
  \(X' \simeq X\). And because \(Y\) is locally \(\V\)-small, we have a
  \(\V\)-valued binary relation \({=_{\V}}\) on \(Y\) such that
  \((y =_{\V} y') \simeq (y = y')\) for every \(y,y' : Y\).
  We now define the \(\V\)-valued equivalence relation \({\approx}\) on \(X'\)
  by \((x \approx x') \colonequiv \pa*{f'(x) =_{\V} f'(x')}\), where
  \(f'\) is the composite \(X' \simeq X \xrightarrow{f} Y\).
  By assumption, the quotient \(X'/{\approx}\) lives in \(\U \sqcup \V\). But it
  is straightforward to work out that \(\image(f)\) is equivalent to this
  quotient. Hence, \(\image(f)\) is \(\pa*{\U \sqcup \V}\)-small, as desired.
\end{proof}

The left-to-right implication of the theorem above is similar
to~\cite[Corollary~5.1]{Rijke2017}, but our theorem generalizes the universe
parameters and restricts to maps into sets. The latter is the reason why the
converse also holds.

\section{Largeness of Complete Posets}
\label{sec:large-posets}
A well-known result of Freyd in classical mathematics says that every
complete small category is a preorder~\cite[Exercise~D of
Chapter~3]{Freyd1964}. In other words, complete categories are necessarily large
and only complete preorders can be small, at least impredicatively.
Predicatively, by contrast, we show that many weakly complete posets (including
directed complete posets, bounded complete posets and sup-lattices) are
necessarily large.
We capture these structures by a technical notion of a \deltacomplete{\V} poset
in~\cref{sec:delta-complete-posets}. In~\cref{sec:nontrivial-and-positive} we
define when such structures are nontrivial and introduce the constructively
stronger notion of positivity. \cref{sec:retract-lemmas} and
\cref{sec:small-completeness-with-resizing} contain the two fundamental
technical lemmas and the main theorems, respectively. Finally, we consider
alternative formulations of being nontrivial and positive that ensure that these
notions are properties rather than data and shows how the main theorems remain
valid, assuming univalence.

\subsection{\texorpdfstring{\(\delta_{\V}\)}{delta\_V}-complete Posets}
\label{sec:delta-complete-posets}
We start by introducing a class of weakly complete posets that we call
\deltacomplete{\V} posets. The notion of a \deltacomplete{\V} poset is a
technical and auxiliary notion sufficient to make our main theorems go
through. The important point is that many familiar structures (dcpos, bounded
complete posets, sup-lattices) are \deltacomplete{\V} posets
(see~\cref{examples-of-delta-complete-posets}).

\begin{defi}[{\deltacomplete{\V} poset,
  \(\delta_{x,y,P}\),
  \(\bigvee \delta_{x,y,P}\)}]
  A \emph{poset} is a type \(X\) with a subsingleton-valued binary relation
  \({\below}\) on \(X\) that is reflexive, transitive and antisymmetric.  It is
  not necessary to require \(X\) to be a set, as this follows from the other
  requirements.
  A poset \((X,{\below})\) is \emph{\(\delta_\V\)-complete} for a
  universe~\(\V\) if for every pair of elements \(x,y : X\) with \(x \below y\)
  and every subsingleton \(P\) in \(\V\), the family
  \begin{align*}
    \delta_{x,y,P} : 1 + P &\to X \\
    \inl(\star) &\mapsto x; \\
    \inr(p) &\mapsto y;
  \end{align*}
  has a supremum \(\bigvee \delta_{x,y,P}\) in \(X\).
\end{defi}
\begin{rem}[Classically, every poset is \deltacomplete{\V}]%
  \label{classically-every-poset-is-delta-complete}
  Consider a poset \((X,\below)\) and a pair of elements \(x \below y\). If
  \(P : \V\) is a decidable proposition, then we can define the supremum of
  \(\delta_{x,y,P}\) by case analysis on whether \(P\) holds or not. For if it
  holds, then the supremum is \(y\), and if it does not, then the supremum is
  \(x\). Hence, if excluded middle holds in \(\V\), then the family
  \(\delta_{x,y,P}\) has a supremum for every \(P : \V\). Thus, if excluded
  middle holds in \(\V\), then every poset (with carrier in any universe) is
  \deltacomplete{\V}.
\end{rem}
The above remark naturally leads us to ask whether the converse also holds,
\ie\ if every poset is \deltacomplete{\V}, does excluded middle in \(\V\)
hold?  As far as we know, we can only get weak excluded middle in \(\V\), as we
will later see in~\cref{Two-is-not-delta-complete}.
This proposition also shows that in the absence of excluded middle, the notion
of \(\delta_{\V}\)-completeness isn't trivial. For now, we focus on the fact
that, also constructively and predicatively, there are many examples of
\deltacomplete{\V} posets.

\begin{exas}\hfill
  \label{examples-of-delta-complete-posets}
  \begin{enumerate}[(i)]
  \item\label{sup-lattices-are-delta-complete} Every \(\V\)-sup-lattice is
    \deltacomplete{\V}. That is, if a poset \(X\) has suprema for all families
    \(I \to X\) with \(I\) in the universe \(\V\), then \(X\) is
    \deltacomplete{\V}.
  \item\label{Omega-is-delta-complete} The \(\V\)-sup-lattice \(\Omega_\V\) is
    \deltacomplete{\V}.  The type \(\Omega_{\V}\) of propositions in \(\V\) is a
    \(\V\)-sup-lattice with the order given by implication and suprema by
    existential quantification. Hence, \(\Omega_{\V}\) is
    \deltacomplete{\V}. Specifically, given propositions \(Q\), \(R\) and \(P\),
    the supremum of \(\delta_{Q,R,P}\) is given by \(Q \vee \pa*{R \times P}\).
  \item\label{powerset-is-delta-complete} The \(\V\)-powerset
    \(\powerset_{\V}(X) \colonequiv X \to \Omega_{\V}\) of a type \(X\) is
    \deltacomplete{\V}. Note that \(\powerset_{\V}(X)\) is another example of a
    \(\V\)-sup-lattice (ordered by subset inclusion and with suprema given by
    unions) and hence \deltacomplete{\V}.
    We will sometimes employ familiar set-theoretic notation when using elements
    of \(\powerset_{\V}(X)\), \eg~given \(A : \powerset_{\V}(X)\), we might
    write \(x \in A\) for the assertion that \(A(x)\) holds.
  \item\label{bounded-complete-are-delta-complete} Every \(\V\)-bounded complete
    poset is \deltacomplete{\V}. That is, if \((X,\below)\) is a poset with
    suprema for all bounded families \(I \to X\) with \(I\) in the universe
    \(\V\), then \((X,\below)\) is \deltacomplete{\V}.
    A family \(\alpha : I \to X\) is bounded if there exists some \(x : X\) with
    \(\alpha(i) \below x\) for every \(i : I\). For example, the family
    \(\delta_{x,y,P}\) is bounded by \(y\).
  \item\label{dcpos-are-delta-complete} Every \(\V\)-directed complete poset
    (dcpo) is \deltacomplete{\V}, since the family \(\delta_{x,y,P}\) is
    directed.  We note that~\cite{deJongEscardo2021} provides a host of examples
    of \(\V\)-dcpos.
  \end{enumerate}
\end{exas}

\subsection{Nontrivial and Positive Posets}
\label{sec:nontrivial-and-positive}
In \cref{classically-every-poset-is-delta-complete} we saw that if we can decide
a proposition \(P\), then we can define \(\bigvee \delta_{x,y,P}\) by case
analysis. What about the converse? That is, if \(\delta_{x,y,P}\) has a supremum
and we know that it equals \(x\) or \(y\), can we then decide \(P\)?  Of course,
if \(x = y\), then \(\bigvee \delta_{x,y,P} = x = y\), so we don't learn
anything about \(P\). But what if add the assumption that \(x \neq y\)? It turns
out that constructively we can only expect to derive decidability of \(\lnot P\)
in that case. This is due to the fact that \(x \neq y\) is a negated
proposition, which is rather weak constructively, leading us to later define
(see~\cref{def:strictly-below}) a constructively stronger notion for elements of
\deltacomplete{\V} posets.

\begin{defi}[Nontriviality]
  A poset \(X\) is \emph{nontrivial} if we have designated \(x,y : X\)
  with \(x \below y\) and \(x \neq y\).
\end{defi}

\begin{lem}\label{delta-sup-weak-em}
  For a nontrivial poset \((X,{\below},x,y)\) and a proposition \(P : \V\), we
  have the following two implications:
  \begin{enumerate}[(i)]
  \item\label{delta-sup-weak-em-1} if the supremum of \(\delta_{x,y,P}\) exists and
    \(x = \bigvee \delta_{x,y,P}\), then \(\lnot P\) is the case;
  \item\label{delta-sup-weak-em-2} if the supremum of \(\delta_{x,y,P}\) exists and
    \(y = \bigvee \delta_{x,y,P}\), then \(\lnot\lnot P\) is the case.
  \end{enumerate}
\end{lem}
\begin{proof}\hfill%
  \begin{enumerate}[(i)]
  \item Suppose that \(x = \bigvee \delta_{x,y,P}\) and assume for a
    contradiction that we have \(p : P\). Then
    \( y \equiv \delta_{x,y,P}(\inr(p)) \below \bigvee \delta_{x,y,P} = x, \)
    which is impossible by antisymmetry and our assumptions that \(x \below y\)
    and \(x \neq y\).
  \item Suppose that \(y = \bigvee \delta_{x,y,P}\) and assume for a
    contradiction that \(\lnot P\) holds. Then
    \(x = \bigvee \delta_{x,y,P} = y\), contradicting our assumption that
    \(x \neq y\). \qedhere
  \end{enumerate}
\end{proof}

\begin{propC}[{\cite[Section~4]{deJongEscardo2021}}]
  \label{Two-is-not-delta-complete}
  If the poset \(\Two\) with exactly two elements \(0 \below 1\) is
  \deltacomplete{\V}, then weak excluded middle in \(\V\) holds.
\end{propC}
\begin{proof}
  Suppose that \(\Two\) were \deltacomplete{\V} and let \(P : \V\) be an
  arbitrary subsingleton. We must show that \(\lnot P\) is decidable. Since
  \(\Two\) has exactly two elements, the supremum \(\bigvee \delta_{0,1,P}\)
  must be \(0\) or \(1\). But then we apply \cref{delta-sup-weak-em} to get
  decidability of \(\lnot P\).
\end{proof}
Combining~\cref{classically-every-poset-is-delta-complete,Two-is-not-delta-complete}
yields that excluded middle implies that every poset is \deltacomplete{\V},
which in turns implies weak excluded middle.
We do not know whether these implications can be reversed.
That the conclusion of the implication in
\cref{delta-sup-weak-em}\ref{delta-sup-weak-em-2} cannot be strengthened to
say that \(P\) is the case is shown by the following observation.
\begin{prop}
  \label{delta-sup-em}
  Recall~\cref{examples-of-delta-complete-posets}, which show that
  \(\Omega_{\V}\) is \deltacomplete{\V}. If for every two propositions \(Q\) and
  \(R\) with \(Q \below R\) and \(Q \neq R\) we have that the equality
  \(R = \bigvee \delta_{Q,R,P}\) in \(\Omega_{\V}\) implies \(P\) for every
  proposition \(P : \V\), then excluded middle in \(\V\) follows.
\end{prop}
\begin{proof}
  Assume the hypothesis in the proposition. We show that \(\lnot\lnot P \to P\)
  for every proposition \(P : \V\), from which excluded middle in \(\V\)
  follows. Let \(P\) be a proposition in~\(\V\) and assume that
  \(\lnot\lnot P\). This yields \(\Zero \neq P\), so by assumption the equality
  \(P = \bigvee \delta_{\Zero,P,P}\) implies~\(P\). But this equality holds,
  because \(\bigvee \delta_{\Zero,P,P} = \Zero \vee (P \times P) = P\), as
  described in
  \cref{examples-of-delta-complete-posets}\ref{Omega-is-delta-complete}.
\end{proof}

Thus, having a pair of elements \(x \below y\) with \(x \neq y\) is rather weak
constructively in that we can only derive \(\lnot\lnot P\) from
\(y = \bigvee\delta_{x,y,P}\).
As promised in the introduction of this section, we now introduce and motivate a
constructively stronger notion.

\begin{defi}[Strictly below, \(x \sbelow y\)]
  \label{def:strictly-below}
  We say that \(x\) is \emph{strictly below} \(y\) in a \deltacomplete{\V} poset
  if \(x \below y\) and, moreover, for every \(z \aboveorder y\) and every
  proposition \(P : \V\), the equality \(z = \bigvee \delta_{x,z,P}\) implies
  \(P\).
\end{defi}
Note that with excluded middle, \(x \sbelow y\) is equivalent to the conjunction
of \(x \below y\) and \(x \neq y\). But constructively, the former is much
stronger, as the following examples and proposition illustrate.

\begin{exas}[Strictly below in \(\Omega_{\V}\) and \(\powerset_{\V}(X)\)]\label{examples-strictly-below}
  \hfill
  \begin{enumerate}[(i)]
  \item Recall from~\cref{examples-of-delta-complete-posets} that
    \(\Omega_{\V}\) is \deltacomplete{\V}. Let \(P : \V\) be an arbitrary
    proposition. Observe that \(\Zero_\V \neq P\) holds precisely when
    \(\lnot\lnot P\) does. However, \(\Zero_\V\) is strictly below \(P\) if and
    only if \(P\) holds.
    More generally, for any two propositions \(Q,P : \V\), we have
    \((Q \below P) \times (Q\neq P)\) if and only if
    \(\lnot Q \times \lnot\lnot P\) holds.
    But, \(Q \sbelow P\) holds if and only if \(\lnot Q \times P\) holds.
  \item Another example (see~\cref{examples-of-delta-complete-posets}) of a
    \deltacomplete{\V}-poset is the powerset \(\powerset_{\V}(X)\) of a type
    \(X : \V\). If we have two subsets \(A \below B\) of \(X\), then
    \(A \neq B\) if and only if
    \(\lnot\pa*{\forall_{x : X}\pa*{x \in B \to x \in A}}\).

    However, if \(A \sbelow B\) and \(y \in A\) is decidable for every
    \(y : X\), then we get the stronger
    \(\exists_{x : X}\pa*{x \in B \times x \not\in A}\).
    For we can take \(P : \V\) to be
    \(\exists_{x : X}\pa*{x \in B \times x \not\in A}\) and observe that
    \(\bigvee\delta_{A,B,P} = B\), because if \(x \in B\), either \(x \in A\) in
    which case \(x \in \bigvee\delta_{A,B,P}\), or \(x \not\in A\) in which case
    \(P\) must hold and \(x \in B = \bigvee\delta_{A,B,P}\).

    Conversely, if we have \(A \below B\) and an element \(x \in B\) with
    \(x \not\in A\), then \(A \sbelow B\). For if \(C \aboveorder B\) is a
    subset and \(P : \V\) a proposition such that \(\bigvee\delta_{A,C,P} = C\),
    then \(x \in C = \bigvee\delta_{A,C,P} = A \cup \{y \in C \mid P\}\), so
  either \(x \in A\) or \(P\) must hold. But \(x \not\in A\) by assumption, so
  \(P\) must be true, proving \(A \sbelow B\).
  \end{enumerate}

\end{exas}

\begin{prop}
  \label{sbelow-below-neq}
  For elements \(x\) and \(y\) of a \deltacomplete{\V} poset, we have that
  \(x \sbelow y\) implies both \(x \below y\) and \(x \neq y\).
  However, if the conjunction of \(x \below y\) and \(x \neq y\) implies
  \(x \sbelow y\) for every \(x,y : \Omega_\V\), then excluded middle in \(\V\)
  holds.
\end{prop}
\begin{proof}
  Note that \(x \sbelow y\) implies \(x \below y\) by definition. Now suppose
  that \(x \sbelow y\). Then the equality
  \(y = \bigvee \delta_{x,y,\Zero_{\V}}\) implies that \(\Zero_{\V}\) holds. But
  if \(x = y\), then this equality holds, so \(x \neq y\), as desired.

  For \(P : \Omega_{\V}\) we observed that \(\Zero_\V \neq P\) is equivalent to
  \(\lnot\lnot P\) and that \(\Zero_\V \sbelow P\) is equivalent to \(P\), so if
  we had \(\pa*{\pa*{x \below y} \times \pa*{x \neq y}} \to x \sbelow y\) in
  general, then we would have \(\lnot\lnot P \to P\) for every proposition \(P\)
  in \(\V\), which is equivalent to excluded middle in \(\V\).
\end{proof}

\begin{lem}\label{sbelow-trans}
  The following transitivity properties hold for all elements \(x\), \(y\) and
  \(z\) of a \deltacomplete{\V} poset:
  \begin{enumerate}[(i)]
  \item if \(x \below y \sbelow z\), then \(x \sbelow z\);
  \item if \(x \sbelow y \below z\), then \(x \sbelow z\).
  \end{enumerate}
\end{lem}
\begin{proof}\hfill
  \begin{enumerate}[(i)]
  \item Assume \(x \below y \sbelow z\), let \(P\) be an arbitrary proposition
    in \(\V\) and suppose that \(z \below w\). We must show that
    \(w = \bigvee \delta_{x,w,P}\) implies \(P\). But \(y \sbelow z\), so we
    know that the equality \(w = \bigvee \delta_{y,w,P}\) implies \(P\). Now
    observe that \(\bigvee \delta_{x,w,P} \below \bigvee \delta_{y,w,P}\), so if
    \(w = \bigvee \delta_{x,w,P}\), then \(w = \bigvee \delta_{y,w,P}\),
    finishing the proof.
  \item Assume \(x \sbelow y \below z\), let \(P\) be an arbitrary proposition
    in \(\V\) and suppose that \(z \below w\). We must show that
    \(w = \bigvee \delta_{x,w,P}\) implies \(P\). But \(x \sbelow y\) and
    \(y \below w\), so this follows immediately. \qedhere
  \end{enumerate}
\end{proof}

\begin{prop}\label{positive-element-equivalent}
  The following are equivalent for an element \(y\) of a \(\V\)-sup-lattice
  \(X\):
  \begin{enumerate}[(i)]
  \item\label{positive-1} the least element of \(X\) is strictly below \(y\);
  \item\label{positive-2} for every family \(\alpha : I \to X\) with
    \(I : \V\), if \(y \below \bigvee \alpha\), then \(I\) is inhabited;
  \item\label{positive-3} there exists some \(x : X\) with \(x \sbelow y\).
  \end{enumerate}
\end{prop}
\begin{proof}
  Write \(\bot\) for the least element of \(X\). By~\cref{sbelow-trans} we have:
  \[
    \bot \sbelow y
    \iff \exists_{x : X}\pa*{\bot \below x \sbelow y}
    \iff \exists_{x : X}\pa*{x \sbelow y},
  \]
  which proves the equivalence of \ref{positive-1} and \ref{positive-3}. It
  remains to prove that \ref{positive-1} and \ref{positive-2} are
  equivalent. Suppose that \(\bot \sbelow y\) and let \(\alpha : I \to X\) with
  \(y \below \bigvee \alpha\). Using \(\bot \sbelow y \below \bigvee \alpha\)
  and~\cref{sbelow-trans}, we have \(\bot \sbelow \bigvee \alpha\). Hence, we
  only need to prove
  \(\bigvee \alpha \below \bigvee \delta_{\bot,\bigvee \alpha,\exists {i :
      I}}\), but
  \(\alpha_j \below \bigvee \delta_{\bot,\bigvee\alpha,\exists {i : I}}\) for
  every \(j : I\), so this is true indeed.
  For the converse, assume that \(y\) satisfies \ref{positive-2}, suppose
  \(z \aboveorder y\) and let \(P : \V\) be a proposition such that
  \(z = \bigvee \delta_{\bot,z,P}\). We must show that \(P\) holds. But notice
  that
  \(y \below z = \bigvee \delta_{\bot,z,P} = \bigvee \pa*{(p : P)\mapsto z}\),
  so \(P\) must be inhabited as \(y\) satisfies~\ref{positive-2}.
\end{proof}
\cref{positive-2}~in~\cref{positive-element-equivalent} says exactly that \(y\)
is a positive element in the sense of~\cite[p.~98]{Johnstone1984}.%
\index{positivity}
Observe that \ref{positive-2} makes sense for any poset, not just
\(\V\)-sup-lattices: we don't need to assume the existence of suprema to
formulate condition~\ref{positive-2}, because we can rephrase
\(y \below \bigvee \alpha\) as ``for every \(x : X\), if \(x\) is an upper bound
of \(\alpha\) and \(x\) is below any other upper bound of \(\alpha\), then
\(y \below x\)''. Similarly, the strictly-below relation makes sense for
any poset.
What \cref{positive-element-equivalent} shows is that the strictly-below
relation generalizes Johnstone's notion of positivity from a \emph{unary}
relation to a \emph{binary} one.
Another binary generalization of positivity in a different direction is that of
a positivity relation in formal
topology~\cite{Sambin2003,CirauloSambin2018,CirauloVickers2016}. For a formal
topology \(S\), one considers a binary relation \(\ltimes\) between \(S\) and
its powerclass. Then \(a \ltimes S\) implies that \(a\) is
positive~\cite[p.~764]{CirauloSambin2018}, while sets of the form
\(\{a \in S \mid a \ltimes U\}\) are thought of as formal closed
subsets~\cite{CirauloVickers2016}.

Looking to strengthen the notion of a nontrivial poset, we make the following
definitions.

\begin{defi}[Positivity; \cf~{\cite[p.~98]{Johnstone1984}}]\hfill
  \begin{enumerate}[(i)]
  \item An element of a \deltacomplete{\V} poset is \emph{positive} if it
    satisfies \cref{positive-element-equivalent}\ref{positive-3}.
  \item A \deltacomplete{\V} poset \(X\) is \emph{positive} if we have
    designated \(x,y : X\) with \(x\) strictly below~\(y\).
  \end{enumerate}
\end{defi}

\begin{exas}[Nontriviality and positivity in \(\Omega_{\V}\) and
  \(\powerset_{\V}(X)\)]\label{examples-nontrivial-positive}\hfill
  \begin{enumerate}[(i)]
  \item Consider an element \(P\) of the \deltacomplete{\V} poset
    \(\Omega_\V\). The pair \(\pa*{\Zero_\V , P}\) witnesses nontriviality of
    \(\Omega_\V\) if and only if \(\lnot\lnot P\) holds, while it witnesses
    positivity if and only if \(P\) holds.
  \item Consider the \(\V\)-powerset \(\powerset_{\V}(X)\) on a type \(X\) as a
    \deltacomplete{\V} poset
    (recall~\cref{examples-of-delta-complete-posets}). We write
    \(\emptyset : \powerset_{\V}(X)\) for the map \(x \mapsto \Zero_\V\).  Say
    that a subset \(A : \powerset_{\V}(X)\) is nonempty if \(A \neq \emptyset\)
    and inhabited if there exists some \(x : X\) such that \(A(x)\) holds.  The
    pair \((\emptyset , A)\) witnesses nontriviality of \(\powerset_\V(X)\) if
    and only if \(A\) is nonempty, while it witnesses positivity if and only if
    \(A\) is inhabited.
  \end{enumerate}
\end{exas}

In domain theory the \emph{way-below} relation is of fundamental importance. It
will be instructive to see how it relates to the strictly-below relation.

\begin{defi}[Way-below relation, compactness; {\cite[Definition~44]{deJongEscardo2021}}]
  Let \(x\) and \(y\) be elements of a \(\V\)-dcpo \(D\).
  \begin{enumerate}[(i)]
  \item We say that \(x\) is \emph{way below} \(y\), written \(x \ll y\), if for
    every directed family \(\alpha : I \to D\) with \(y \below \bigvee \alpha\),
    there exists \(i : I\) such that \(x \below \alpha_i\) already.
  \item An element \(x\) is said to be \emph{compact} if it is way below itself.
  \end{enumerate}
\end{defi}

\begin{prop}
  If \(x \below y\) are unequal elements of a \(\V\)-dcpo \(D\) and \(y\) is
  compact, then \(x \sbelow y\) without the need to assume excluded middle.
  In particular, a compact element \(x\) of a \(\V\)-dcpo with a least element
  \(\bot\) is positive if and only if \(x \neq \bot\).
\end{prop}
\begin{proof}
  Suppose that \(x \below y\) are unequal and that \(y\) is compact. We are to
  show that \(x \sbelow y\). So assume we have \(z \aboveorder y\) and a
  proposition \(P : \V\) such that \(y \below z = \bigvee\delta_{x,z,P}\). By
  compactness of \(y\), there exists \(i : \One + P\) such that
  \(y \below \delta_{x,z,P}(i)\) already. But \(i\) can't be equal to
  \(\inl(\star)\), since \(x \neq y\) is assumed. Hence, \(i = \inr(p)\) and
  \(P\) must hold.
\end{proof}

Note that \(x \sbelow y\) does not imply \(x \ll y\) in general, because with
excluded middle, \(x \sbelow y\) is simply the conjunction of \(x \below y\) and
\(x \neq y\), which does not imply \(x \ll y\) in general.
Also, the conjunction of \(x \ll y\) and \(x \neq y\) does not imply
\(x \sbelow y\), as far as we know.

We end this section by summarizing why we consider the strictly-below relation
to be suitable in our constructive framework.
First of all, \(x \sbelow y\) coincides with \((x \below y) \times (x \neq y)\)
in the presence of excluded middle, so it is compatible with classical logic.
Secondly, we've seen in~\cref{examples-strictly-below} that the strictly-below
relation works well in the poset of truth values and in powersets, yielding
familiar constructive strengthenings.
Thirdly, the strictly-below relation generalizes Johnstone's notion of
positivity from a unary to a binary relation.
And finally, as we will see shortly, the derived notion of positive poset is
exactly what we need to derive \(\Omegaresizingalt{\V}\) rather than the weaker
\(\Omeganotnotresizingalt{\V}\) in \cref{positive-impredicativity}.

\subsection{Retract Lemmas}\label{sec:retract-lemmas}
We show that the type of propositions in \(\V\) is a retract of any positive
\deltacomplete{\V} poset and that the type of \(\lnot\lnot\)-stable
propositions in \(\V\) is a retract of any nontrivial \deltacomplete{\V} poset.

\begin{defi}[\(\Delta_{x,y}\)]
  For a nontrivial \deltacomplete{\V} poset \((X,\below,x,y)\), we define the
  map \(\Delta_{x,y} : \Omega_{\V} \to X\) by the assignment
  \(P \mapsto \bigvee \delta_{x,y,P}\).
\end{defi}
We will often omit the subscripts in \(\Delta_{x,y}\) when it is clear from the
context.

\begin{defi}[Locally smallness]
  A \deltacomplete{\V} poset \((X,\below)\) is \emph{locally small} if its order
  has \(\V\)-small values, \ie\ we have
  \({\below_{\V}} : X \to X \to \V\) with
  \(\pa*{x \below y} \simeq \pa*{x \below_{\V} y}\) for every \(x,y : X\).
\end{defi}

\begin{exas}\hfill
  \begin{enumerate}[(i)]
  \item The \(\V\)-sup-lattices \(\Omega_{\V}\) and \(\powerset_\V(X)\) (for
    \(X : \V\)) are locally small.
  \item All examples of \(\V\)-dcpos in~\cite{deJongEscardo2021} are locally
    small.
  \end{enumerate}
\end{exas}

\begin{lem}\label{nontrivial-retract}
  A locally small \deltacomplete{\V} poset \((X,\below)\) is nontrivial,
  witnessed by elements \(x \below y\), if~and only if the composite
  \(\Omeganotnot{\V} \hookrightarrow \Omega_{\V} \xrightarrow{\Delta_{x,y}} X\)
  is a section.
\end{lem}
\begin{proof}
  Suppose first that \((X,\below,x,y)\) is nontrivial and locally small. We define
  \begin{align*}
    r : X &\to \Omeganotnot{\V} \\
    z &\mapsto z \not\below_{\V} x.
  \end{align*}
  Note that negated propositions are \(\lnot\lnot\)-stable, so \(r\) is
  well-defined. Let \(P : \V\) be an arbitrary \(\lnot\lnot\)-stable
  proposition. We want to show that \(r (\Delta_{x,y}(P)) = P\). By propositional
  extensionality, establishing logical equivalence suffices.
  Suppose first that \(P\) holds. Then
  \(\Delta_{x,y}(P) \equiv \bigvee \delta_{x,y,P} = y\), so
  \(r(\Delta_{x,y}(P)) = r(y) \equiv \pa*{y \not\below_{\V} x}\) holds by
  antisymmetry and our assumptions that \(x \below y\) and \(x \neq y\).
  Conversely, assume that \(r(\Delta_{x,y}(P))\) holds, \ie\ that we have
  \(\bigvee \delta_{x,y,P} \not\below_{\V} x\). Since \(P\) is
  \(\lnot\lnot\)-stable, it suffices to derive a contradiction from~\(\lnot
  P\). So assume~\(\lnot P\). Then \(x = \bigvee \delta_{x,y,P}\), so
  \(r(\Delta_{x,y}(P)) = r(x) \equiv x \not\below_{\V} x\), which is false by
  reflexivity.

  For the converse, assume that
  \(\Omeganotnot{\V} \hookrightarrow \Omega_{\V} \xrightarrow{\Delta_{x,y}} X\)
  has a retraction \(r : \Omeganotnot{\V} \to X\). Then
  \(\Zero_{\V} = r(\Delta_{x,y}(\Zero_{\V})) = r(x)\) and
  \(\One_{\V} = r(\Delta_{x,y}(\One_{\V})) = r(y)\),
  where we used that \(\Zero_{\V}\) and \(\One_{\V}\) are \(\lnot\lnot\)-stable.
  Since \(\Zero_{\V} \neq \One_{\V}\), we get \(x \neq y\), so
  \((X,\below,x,y)\) is nontrivial, as desired.
\end{proof}
The appearance of the double negation in the above lemma is due to the
definition of nontriviality. If we instead assume a positive poset \(X\), then
we can exhibit all of \(\Omega_{\V}\) as a retract of \(X\).
\begin{lem}\label{positive-retract}
  A locally small \deltacomplete{\V} poset \((X,\below)\) is positive, witnessed
  by elements \(x \sbelow y\), if~and only if for every \(z \aboveorder y\), the
  map \(\Delta_{x,z} : \Omega_{\V} \to X\) is a section.
\end{lem}
\begin{proof}
  Suppose first that \((X,\below,x,y)\) is positive and locally small and let
  \(z \aboveorder y\) be arbitrary. We define
  \begin{align*}
    r_z : X &\mapsto \Omega_{\V} \\
    w &\mapsto z \below_{\V} w.
  \end{align*}
  Let \(P : \V\) be arbitrary proposition. We want to show that
  \(r_z(\Delta_{x,z}(P)) = P\). Because of propositional extensionality, it
  suffices to establish a logical equivalence between \(P\) and
  \(r_z(\Delta_{x,z}(P))\).
  Suppose first that \(P\) holds. Then \(\Delta_{x,z}(P) = z\), so
  \(r_z(\Delta_{x,z}(P)) = r_z(z) \equiv \pa*{z \below_{\V} z}\) holds as well
  by reflexivity.
  Conversely, assume that \(r_z(\Delta_{x,z}(P))\) holds, \ie\ that we have
  \(z \below_{\V} \bigvee \delta_{x,z,P}\). Since
  \({\bigvee \delta_{x,z,P} \below z}\) always holds, we get
  \(z = \bigvee \delta_{x,z,P}\) by antisymmetry. But by assumption
  and~\cref{sbelow-trans}, the element \(x\) is strictly~below~\(z\), so \(P\)
  must hold.

  For the converse, assume that for every \(z \aboveorder y\), the map
  \(\Delta_{x,z} : \Omega_{\V} \to X\) has a retraction
  \(r_z : X \to \Omega_{\V}\). We must show that the equality
  \(z = \Delta_{x,z}(P)\) implies \(P\) for every \(z \aboveorder y\) and
  proposition \(P : \V\). Assuming \(z = \Delta_{x,z}(P)\), we have
  \(\One_{\V} = r_z(\Delta_{x,z}(\One_{\V})) = r_z(z) = r_z(\Delta_{x,z}(P)) =
  P\), so \(P\) must hold indeed. Hence, \((X,\below,x,y)\) is positive, as
  desired.
\end{proof}

\subsection{Small Completeness with Resizing}
\label{sec:small-completeness-with-resizing}
We present our main theorems here, which show that, constructively and
predicatively, nontrivial \deltacomplete{\V} posets are necessarily large and
necessarily lack decidable equality.

\begin{defi}[Smallness]
  A \deltacomplete{\V} poset is \emph{small} if it is locally small and its
  carrier is \(\V\)-small.
\end{defi}

\begin{thm}\label{nontrivial-impredicativity}\label{positive-impredicativity}\hfill
  \begin{enumerate}[(i)]
  \item\label{nontrivial-impredicativity-1} There is a nontrivial small
    \deltacomplete{\V} poset if and only if \(\Omeganotnotresizingalt{\V}\)
    holds.
  \item\label{positive-impredicativity-2} There is a positive small
    \deltacomplete{\V} poset if and only if \(\Omegaresizingalt{\V}\) holds.
  \end{enumerate}
\end{thm}
\begin{proof}\hfill
  \begin{enumerate}[(i)]
  \item Suppose that \((X,\below,x,y)\) is a nontrivial small \deltacomplete{\V}
    poset. By \cref{nontrivial-retract}, we can exhibit \(\Omeganotnot{\V}\) as
    a retract of \(X\). But \(X\) is \(\V\)-small by assumption, so
    by~\cref{is-small-retract} the type \(\Omeganotnot{\V}\) is \(\V\)-small as
    well.
    For the converse, note that
    \(\pa*{\Omeganotnot{\V},\to,\Zero_{\V},\One_{\V}}\) is a nontrivial locally
    small \(\V\)-sup-lattice with \(\bigvee \alpha\) given by
    \(\lnot\lnot\exists_{i : I}\alpha_i\). And if we assume
    \(\Omeganotnotresizingalt{\V}\), then it is small.
  \item Suppose that \((X,\below,x,y)\) is a positive small poset. By
    \cref{positive-retract}, we can exhibit \(\Omega_{\V}\) as a retract of
    \(X\). But \(X\) is \(\V\)-small by assumption, so
    by~\cref{is-small-retract} the type \(\Omega_{\V}\) is \(\V\)-small as well.
    For the converse, note that \(\pa*{\Omega_{\V},\to,\Zero_\V,\One_\V}\) is a
    positive locally small \(\V\)-sup-lattice. And if we assume
    \(\Omegaresizingalt{\V}\), then it is small. \qedhere
  \end{enumerate}
\end{proof}

\begin{lemC}[{\cite[{\mkTTurl{TypeTopology.DiscreteAndSeparated}}]{TypeTopology}}]\hfill
  \label{equality-retract}
  \begin{enumerate}[(i)]
  \item Types with decidable equality are closed under retracts.
  \item Types with \(\lnot\lnot\)-stable equality are closed under retracts.
  \end{enumerate}
\end{lemC}

\begin{exas}[Types with \(\lnot\lnot\)-stable equality]%
  \label{types-with-not-not-stable-equality}
  The simple types \(\Nat\), \({\Nat \to \Nat}\), \({\Nat \to \Nat \to \Nat}\),
  etc.~{\cite[{\mkTTurl{TypeTopology.SimpleTypes}}]{TypeTopology}}, and the type
  of Dedekind real
  numbers~{\cite[{\mkTTurl{Various.Dedekind}}]{TypeTopology}} all have
  \(\lnot\lnot\)-stable equality, as does the type \(\Omeganotnot{\U}\) of
  \(\lnot\lnot\)\nobreakdash-stable propositions in any universe \(\U\).
\end{exas}

\begin{thm}\label{nontrivial-weak-em}
  There is a nontrivial locally small \deltacomplete{\V} poset with decidable
  equality if and only if weak excluded middle in \(\V\) holds.
\end{thm}
\begin{proof}
  Suppose that \((X,\below,x,y)\) is a nontrivial locally small
  \deltacomplete{\V} poset with decidable equality. Then by
  Lemmas~\ref{nontrivial-retract} and \ref{equality-retract}, the type
  \(\Omeganotnot{\V}\) must have decidable equality too. But negated
  propositions are \(\lnot\lnot\)-stable, so this yields weak excluded middle in
  \(\V\). For the converse, note that
  \(\pa*{\Omeganotnot{\V},\to,\Zero_\V,\One_\V}\) is a nontrivial locally small
  \(\V\)-sup-lattice that has decidable equality if and only if weak excluded
  middle in \(\V\) holds.
\end{proof}

\begin{thm}\label{positive-em} The following are equivalent:
  \begin{enumerate}[(i)]
  \item\label{positive-em-1} there is a positive locally small
    \deltacomplete{\V} poset with \(\lnot\lnot\)-stable equality;
  \item\label{positive-em-2} there is a positive locally small
    \deltacomplete{\V} poset with decidable equality;
  \item\label{positive-em-3} excluded middle in \(\V\) holds.
  \end{enumerate}
\end{thm}
\begin{proof}
  Note that
  \(\textup{\ref{positive-em-2}}\Rightarrow\textup{\ref{positive-em-1}}\), so we
  are left to show that
  \(\textup{\ref{positive-em-3}}\Rightarrow\textup{\ref{positive-em-2}}\) and
  that
  \(\textup{\ref{positive-em-1}}\Rightarrow\textup{\ref{positive-em-3}}\). For
  the first implication, note that \(\pa*{\Omega_{\V},\to,\Zero_\V,\One_\V}\) is
  a positive locally small \(\V\)-sup-lattice that has decidable equality if and
  only if excluded middle in \(\V\) holds.
  To see that \ref{positive-em-1}~implies~\ref{positive-em-3}, suppose that
  \((X,\below,x,y)\) is a positive locally small \deltacomplete{\V} poset with
  \(\lnot\lnot\)-stable equality. Then by
  Lemmas~\ref{positive-retract}~and~\ref{equality-retract} the type
  \(\Omega_{\V}\) must have \(\lnot\lnot\)-stable equality. But this implies
  that \(\lnot\lnot P \to P\) for every proposition \(P\) in \(\V\) which is
  equivalent to excluded middle in \(\V\).
\end{proof}

In particular, \cref{positive-em}\ref{positive-em-1} shows that, constructively,
none of the types from~\cref{types-with-not-not-stable-equality} can be equipped
with the structure of a positive \deltacomplete{\V} poset.
In particular, we cannot expect the type of Dedekind reals to be a positive
bounded complete poset.

Lattices, bounded complete posets and dcpos are necessarily large and
necessarily lack decidable equality in our predicative constructive setting.
More precisely:
\begin{cor}\hfill
  \begin{enumerate}[(i)]
  \item There is a nontrivial small \(\V\)-sup-lattice (or \(\V\)-bounded complete
    poset or \(\V\)-dcpo) if~and~only~if \(\Omeganotnotresizingalt{\V}\) holds.
  \item There is a positive small \(\V\)-sup-lattice (or \(\V\)-bounded complete
    poset or \(\V\)-dcpo) if~and~only~if \(\Omegaresizingalt{\V}\) holds.
  \item There is a nontrivial locally small \(\V\)-sup-lattice (or
    \(\V\)-bounded complete poset or \(\V\)-dcpo) with decidable equality if and
    only if weak excluded middle in \(\V\) holds.
  \item There is a positive locally small \(\V\)-sup-lattice (or \(\V\)-bounded
    complete poset or \(\V\)-dcpo) with decidable equality if and only if
    excluded middle in \(\V\) holds.
  \end{enumerate}
\end{cor}

The above notions of non-triviality and positivity are data rather than
property. Indeed, a nontrivial poset \((X,\below)\) is (by definition) equipped
with two designated points \(x,y : X\) such that \(x \below y\) and
\(x \neq y\). It is natural to wonder if the propositionally truncated versions
of these two notions yield the same conclusions. We show that this is indeed the
case if we assume univalence. The need for the univalence assumption comes from
the fact that smallness is a property precisely if univalence holds, as shown in
Propositions~\ref{is-small-is-prop}~and~\ref{is-small-univalence}.

\begin{defi}[Nontrivial/positive in an unspecified way]
  A poset \((X,\below)\) is \emph{nontrivial in an unspecified way} if there
  exist some elements \(x,y : X\) such that \(x \below y\) and \(x \neq y\),
  \ie\
  \(\exists_{x,y: X}\pa*{\pa*{x \below y} \times \pa*{x \neq y}}\).
  Similarly, we can define when a poset is \emph{positive in an unspecified way}
  by truncating the notion of positivity.
\end{defi}

\begin{thm}
  Suppose that the universes \(\V\) and \(\V^+\) are univalent.
  \begin{enumerate}[(i)]
  \item\label{unspecified-1} There is a small \deltacomplete{\V} poset that is
    nontrivial in an unspecified way if and only if
    \(\Omeganotnotresizingalt{\V}\) holds.
  \item\label{unspecified-2} There is a small \deltacomplete{\V} poset that is
    positive in an unspecified way if and only if \(\Omegaresizingalt{\V}\)
    holds.
  \end{enumerate}
\end{thm}
\begin{proof}\hfill%
  \begin{enumerate}[(i)]
  \item Suppose that \((X,\below)\) is a \deltacomplete{\V} poset that is
    nontrivial in an unspecified way.  By~\cref{is-small-is-prop} and univalence
    of \(\V\) and \(\V^+\), the type \(``\Omeganotnot{\V}\issmall{\V}"\) is a
    proposition. By the universal property of the propositional truncation, in
    proving that the type \(\Omeganotnot{\V}\) is \(\V\)-small we can therefore
    assume that are given points \(x,y : X\) with \(x \below y\) and
    \(x \neq y\). The result then follows
    from~\cref{nontrivial-impredicativity}.
  \item By reduction to item~\ref{positive-impredicativity-2} of
    \cref{positive-impredicativity}. \qedhere
  \end{enumerate}
\end{proof}
Similarly, we can prove the following theorems by reduction to
Theorems~\ref{nontrivial-weak-em}~and~\ref{positive-em}.

\begin{thm}\hfill
  \begin{enumerate}[(i)]
  \item There is a locally small \deltacomplete{\V} poset with decidable
    equality that is nontrivial in an unspecified way if and only if weak
    excluded middle in \(\V\) holds.
  \item There is a locally small \deltacomplete{\V} poset with decidable
    equality that is positive in an unspecified way if and only if excluded
    middle in \(\V\) holds.
  \end{enumerate}
\end{thm}


\section{Maximal Points and Fixed Points}\label{sec:maximal-and-fixed-points}

As is well known, in impredicative mathematics, a poset has suprema of all
subsets if and only if it has infima of all subsets.
Perhaps counter-intuitively, this ``duality'' theorem can be proved
predicatively. However, in the absence of impredicativity, it is not possible to
fulfil its hypotheses when trying to apply it, because there are no nontrivial
examples.

To explain this, we first have to make the statement of the duality theorem
precise. A~single universe formulation is ``every \(\V\)-small
\(\V\)-sup-lattice has all infima of families indexed by types in \(\V\)''.
The usual proof, adapted from subsets to families, shows that this formulation
is predicatively provable, but in our predicative setting
\cref{nontrivial-impredicativity} tells us that there are no nontrivial examples
to apply it to.

It is natural to wonder whether the single universe formulation can be
generalized to \emph{locally} small \(\V\)-sup-lattices (with necessarily large
carriers), resulting in a predicatively useful result.
However, as one of the anonymous reviewers pointed out that this generalization
gives rise to a false statement and suggested the ordinals as a counterexample
in a set-theoretic setting: it is a class with suprema for all subsets but has
no greatest element.
This led us to prove~(\cref{sec:small-suprema-of-ordinals}) in our
type-theoretic context that the locally small, but large type of ordinals in a
univalent universe~\(\V\) is a \(\V\)-sup-lattice. But this is not a
\(\V\)-inf-lattice, because the unique family indexed by the empty type does not
have a greatest lower bound since the type of ordinals has no greatest element.

Similarly, consider a generalized formulation of Tarski's
theorem~\cite{Tarski1955} that allows for multiple universes, \ie\ we define
\(\text{\emph{Tarski's-Theorem}}_{\V,\U,\T}\) as the assertion that every
monotone endofunction on a \(\V\)-sup-lattice with carrier in a universe \(\U\)
and order taking values in a universe \(\T\) has a greatest fixed point.
Then \(\Tarski{\V}{\V}{\V}\) corresponds to the original formulation and,
moreover, is provable predicatively, but not useful predicatively because
Theorem 4.24 shows that its hypotheses can only be fulfilled for trivial posets.
On~the other hand, \(\Tarski{\V}{\V^+}{\V}\) is provably false because the
identity map on the \(\V\)-sup-lattice of ordinals in \(\V\) is a
counterexample.
Analogous considerations can be made for a lemma due to
Pataraia~\cite{Pataraia1997,Escardo2003} saying that every dcpo has a greatest
monotone inflationary endofunction.

\subsection{A Predicative Counterexample}\label{sec:predicative-counterexample}

Because the type of ordinals in \(\V\) is not \(\V\)-small even impredicatively,
the above does not rule out the possibility that a \(\V\)-sup-lattice \(X\) has
all \(\V\)-infima provided \(X\) is \(\V\)-small impredicatively.
To address this, we present an example of a \(\V\)-sup-lattice, parameterized by
a proposition, that is \(\V\)-small impredicatively, but predicatively does not
necessarily have a maximal element. In particular, it need not have a greatest
element or all \(\V\)-infima.

\begin{defi}[Lifting, \cf~\cite{EscardoKnapp2017}]\label{lifting-of-prop}
  Fix a proposition \(P_\U\) in a universe \(\U\). Lifting \(P_\U\) with respect
  to a universe \(\V\) is defined by
  \[
    \lifting_{\V}\pa*{P_\U} \colonequiv \sum_{Q : \Omega_\V} \pa*{Q \to P_\U}.
  \]
\end{defi}
This is a subtype of \(\Omega_\V\) (the map
\(\fst : \lifting_{\V}\pa*{P_\U} \to \Omega_\V\) is an embedding) and it is
closed under \(\V\)-suprema (in particular, it contains the least~element).

\begin{exas}\hfill
  \begin{enumerate}[(i)]
  \item If \(P_\U \colonequiv \Zero_\U\), then
    \( \lifting_{\V}(P_\U) \simeq \pa*{\Sigma_{Q : \Omega_\V} \lnot Q} \simeq
    \pa*{\Sigma_{Q : \Omega_\V} \pa*{Q = \Zero_{\V}}} \simeq \One \).
  \item If \(P_\U \colonequiv \One_\U\), then
    \( \lifting_{\V}(P_\U) \equiv \pa*{\Sigma_{Q : \Omega_\V} \pa*{Q \to \One_\U}}
    \simeq \Omega_\V \).
  \end{enumerate}
\end{exas}
What makes \(\lifting_{\V}(P_\U)\) useful is the following observation.
\begin{lem}\label{maximal-iff-resize}
  Suppose that the poset \(\lifting_{\V}(P_\U)\) has a maximal element
  \(Q : \Omega_\V\). Then \(P_\U\) is equivalent to \(Q\), which is the
  greatest element of \(\lifting_{\V}(P_\U)\). In particular,
  \(P_\U\) is \(\V\)-small.
  Conversely, if \(P_\U\) is equivalent to a proposition \(Q : \Omega_\V\), then
  \(Q\) is the greatest element of~\(\lifting_{\V}(P_\U)\).
\end{lem}
\begin{proof}
  Suppose that \(\lifting_{\V}(P_\U)\) has a maximal element \(Q :
  \Omega_\V\). We wish to show that \(Q \simeq P_\U\). By definition of
  \(\lifting_{\V}(P_\U)\), we already have that \(Q \to P_\U\). So only the
  converse remains. Therefore suppose that \(P_\U\) holds. Then, \(\One_\V\) is
  an element of \(\lifting_{\V}(P_\U)\). Obviously \(Q \to 1_\V\), but \(Q\) is
  maximal, so actually \(Q = 1_\V\), that is, \(Q\) holds, as
  desired. Thus, \(Q \simeq P_\U\). It~is then straightforward to see that \(Q\)
  is actually the greatest element of \(\lifting_{\V}(P_\U)\), since
  \(\lifting_{\V}(P_\U) \simeq \Sigma_{Q' : \Omega_\V}(Q' \to Q)\).
  For the converse, assume that \(P_\U\) is equivalent to a proposition
  \(Q : \Omega_\V\). Then, as before,
  \(\lifting_{\V}(P_\U) \simeq \Sigma_{Q' : \Omega_\V}(Q' \to Q)\), which shows
  that \(Q\) is indeed the greatest element of~\(\lifting_{\V}(P_\U)\).
\end{proof}

\begin{cor}
  The \(\V\)-sup-lattice \(\lifting_{\V}(P_{\U})\) has all \(\V\)-infima if and
  only if \(P_{\U}\) is \(\V\)-small.
\end{cor}
\begin{proof}
  Suppose first that \(\lifting_{\V}(P_{\U})\) has all \(\V\)-infima. Then it
  must have a infimum for the empty family
  \(\Zero_{\V} \to \lifting_{\V}(P_{\U})\). But this infimum must be the greatest
  element of \(\lifting_{\V}(P_{\U})\). So by \cref{maximal-iff-resize} the
  proposition \(P_{\U}\) must be \(\V\)-small.

  Conversely, suppose that \(P_{\U}\) is equivalent to a proposition \(Q : \V\).
  Then the infimum of a family \(\alpha : I \to \lifting_{\V}(P_{\U})\) with
  \(I : \V\) is given by \(\pa*{Q \times \Pi_{i : I} \alpha_i} : \V\).
\end{proof}

In~\cite{FSCDversion} we used~\cref{maximal-iff-resize} to conclude that a
version of Zorn's lemma that says that every pointed dcpo has a maximal element
is predicatively unavailable, as \(\lifting_{\V}\pa*{P_{\U}}\) is a pointed
\(\V\)-dcpo, but has a maximal element if and only if \(P_{\U}\) is \(\V\)-small.
But, as in our discussion above of the duality theorem and Tarski's theorem, we
must pay attention to the universes here.
Zorn's lemma restricted to \(\V\)-small \(\V\)-sup-lattices is, assuming
excluded middle~\cite{Bell1997}, equivalent to the axiom of choice, as usual.
Disregarding its constructive status for a moment, the predicative issue is
that there are no nontrivial \(\V\)-small \(\V\)-sup-lattices
(\cref{nontrivial-impredicativity}).
But the generalization of Zorn's lemma to \emph{locally} small
\(\V\)-sup-lattices is false (even if we assume the axiom of choice and hence,
excluded middle), because the \(\V\)-sup-lattice of ordinals in \(\V\), having
no maximal element, is a counterexample.

\subsection{Small Suprema of Small Ordinals in Univalent Foundations}
\label{sec:small-suprema-of-ordinals}

We now show that the ordinal \(\Ord_{\V}\) of ordinals in a fixed
\emph{univalent} universe \(\V\) has suprema for all families indexed by types
in~\(\V\) and that it has no maximal element.
The latter is implied by~\cite[Lemma~10.3.21]{HoTTBook}, but we were not able to
find a proof of the former in the literature:
Theorem~9 of~\cite{KrausForsbergXu2021} only proves \(\Ord_{\V}\) to have joins
of increasing sequences, while \cite[Lemma~10.3.22]{HoTTBook} shows that every
family indexed by a type in \(\V\) has some upper bound, but does not prove it
to be the least (although least upper bounds are required for
\cite[Exercise~10.17(ii)]{HoTTBook}).
We present two proofs: one based on~\cite[Lemma~10.3.22]{HoTTBook} using small
set quotients and an alternative one using small images.

Following~\cite[Section~10.3]{HoTTBook}, we define an ordinal to be a type
equipped with a proposition-valued, transitive, extensional and (inductive)
well-founded relation.
In \cite{HoTTBook} the underlying type of an ordinal is required to be a set,
but this actually follows from the other axioms,
see~\cite[{\mkTTurllong{Ordinals.Type}{underlying-type-is-set}}]{TypeTopology}.
The type of ordinals, denoted by \(\Ord_{\V}\), in a given \emph{univalent}
universe \(\V\) can itself be equipped with such a
relation~\cite[Theorem~10.3.20]{HoTTBook} and thus is an ordinal again. However,
it is not an ordinal in \(\V\), but rather in the next universe \(\V^+\), and
this is necessary, because it is contradictory for \(\Ord_{\V}\) to be
isomorphic to an ordinal in \(\V\),
see~\cite{TypeTopologyBuraliForti}.

Before we prove that \(\Ord_{\V}\) has \(\V\)-suprema, we need to recall a few
facts.
The well-order on \(\Ord_{\V}\) is given by: \(\alpha \prec \beta\) if and only
if we can find a (necessarily) unique \(y : \beta\) such that
\(\alpha\)~and~\(\beta \initseg y\) are isomorphic ordinals. Here
\(\beta \initseg y\) denotes the ordinal of elements \(b : \beta\) satisfying
\(b \prec y\).

\begin{lemC}[{\cite[{\mkTTurl{Ordinals.OrdinalOfOrdinals}}]{TypeTopology}}]%
  \label{initseg-order}
  For every two points \(x\) and \(y\) of an ordinal \(\alpha\), we have
  \(x \prec y\) in \(\alpha\) if and only if
  \(\alpha \initseg x \prec \alpha \initseg y\) as ordinals.
\end{lemC}
\begin{proof}
  If \(x \prec y\), then we can consider \(\alpha \initseg y \initseg x\) which
  is easily seen to be isomorphic to \(\alpha \initseg x\), so that
  \(\alpha \initseg x \prec \alpha \initseg y\).
  Conversely, if \(\alpha \initseg x \prec \alpha \initseg y\), then
  \(\alpha \initseg x\) is isomorphic to \(\alpha \initseg y \initseg z\) for
  some unique \(z \prec y\).
  But now \(\alpha \initseg x\) and \(\alpha \initseg z\) are isomorphic which
  implies that \(x = z \prec y\).
\end{proof}

\begin{defi}[Simulation, {\cite[Section~10.3]{HoTTBook}}]
  A \emph{simulation} between two ordinals \(\alpha\)~and~\(\beta\) is a map
  \(f : \alpha \to \beta\) satisfying the following conditions:
  \begin{enumerate}[(i)]
  \item for every \(x,y : \alpha\), if \(x \prec y\), then \(f(x) \prec f(y)\);
  \item for every \(x : \alpha\) and \(y : \beta\), if \(y \prec f(x)\), then we
    can find a (necessarily unique) \(x' : \alpha\) such that \(x' \prec x\) and
    \(f(x') = y\).
  \end{enumerate}
\end{defi}

\begin{lemC}[{\cite[{\mkTTurl{Ordinals.OrdinalOfOrdinals}}]{TypeTopology}}]%
  \label{Ord-order}
  For ordinals \(\alpha\) and \(\beta\), the following are equivalent:
  \begin{enumerate}[(i)]
  \item we can find a (necessarily unique) simulation from \(\alpha\) to
    \(\beta\);
  \item for every ordinal \(\gamma\), if \(\gamma \prec \alpha\), then
    \(\gamma \prec \beta\).
  \end{enumerate}
  We write \(\alpha \preceq \beta\) in case the equivalent conditions above
  hold.
\end{lemC}
\begin{proof}
  Given a simulation \(f : \alpha \to \beta\) and an ordinal
  \(\gamma \prec \alpha\), we have \(x : \alpha\) such that
  \(\gamma\)~and~\(\alpha \initseg x\) are isomorphic.
  We claim that \(\alpha \initseg x\) and \(\beta \initseg f(x)\) are
  isomorphic, which entails \(\gamma < \beta\), as desired.
  The forward direction of the isomorphism is given by \(f\), while in the other
  direction we map \(y \prec f(x)\) to the unique \(x' : \alpha\) with
  \(f(x') = y\) given by the fact that \(f\) is a simulation.

  Conversely, if \(\gamma \prec \alpha\) implies \(\gamma \prec \beta\), then
  for every \(x : \alpha\), we have a unique \(y : \beta\) such that
  \(\alpha \initseg x\) and \(\beta \initseg y\) are isomorphic.
  This defines a map \(f : \alpha \to \beta\) which is easily seen to be
  monotone.
  Moreover, if \(y \prec f(x)\), then
  \(\beta \initseg y \prec \alpha \initseg x\), so that we get \(x' \prec x\)
  with \(y = f(x')\), and \(f\) is thus a simulation.
\end{proof}

Recall from~\cref{existence-of-small-set-quotients} what it means to have small
set quotients. If these are available, then the type of ordinals has all small
suprema.

\begin{thm}[Extending~{\cite[Lemma~10.3.22]{HoTTBook}}]
  Assuming small set quotients, the large ordinal \(\Ord_{\V}\) has suprema of
  families indexed by types in \(\V\).
\end{thm}
\begin{proof}[Proof~{\cite{TypeTopologyOrdinalSuprema}}]
  Given \(\alpha : I \to \Ord_{\V}\), define \(\hat{\alpha}\) as the quotient of
  \(\Sigma_{i : I}\,\alpha_i\) by the \(\V\)-valued equivalence relation
  \({\approx}\) where \((i,x) \approx (j,y)\) if and only if \(\alpha_i \initseg x\)
  and \(\alpha_j \initseg y\) are isomorphic as ordinals.
  By our assumption, the quotient \(\hat\alpha\) lives in~\(\V\).
  Next, \cite[Lemma~10.3.22]{HoTTBook} tells us that \(\pa*{\hat{\alpha},\prec}\)
  with
  \[
    [(i,x)] \prec [(j,y)] \colonequiv (\alpha_i \initseg x) \prec (\alpha_j \initseg y).
  \]
  is an ordinal that is an upper bound of \(\alpha\).
  So we show that \(\hat{\alpha}\) is a lower bound of upper bounds of
  \(\alpha\). To this end, suppose that \(\beta : \Ord_\V\) is such that
  \(\alpha_i \preceq \beta\) for every \(i : I\). In light of~\cref{Ord-order},
  this assumption yields two things:
  \begin{enumerate}[(1)]
  \item\label{assum1} for every \(i : I\) and \(x : \alpha_i\) there exists a
    unique \(b_i^x : \beta\) such that
    \(\alpha_i \initseg x = \beta \initseg b_i^x\);
  \item\label{assum2} for every \(i : I\), a simulation
    \(f_i : \alpha_i \to \beta\) such that for every \(x : \alpha_i\), we have
    \(f_i(x) = b_i^x\).
  \end{enumerate}
  We are to prove that \(\hat{\alpha} \preceq \beta\). We start by defining
  \begin{align*}
    f : \pa*{\Sigma_{i : I}\,\alpha_i} &\to \beta \\
    (i,x) &\mapsto b_i^x
  \end{align*}
  Observe that \(f\) respects \(\approx\), for if \((i,x) \approx (j,y)\), then by
  univalence,
  \[(\beta \initseg b_i^x) = (\alpha_i \initseg x) = (\alpha_j \initseg y)
    = (\beta \initseg b_j^y),\]
  so \(b_i^x = b_j^y\) by uniqueness of \(b_i^x\).
  Thus, \(f\) induces a map \(\hat{f} : \hat{\alpha} \to \beta\) satisfying the equality
  \(\hat{f}([(i,x)]) = f(i,x)\) for every \((i,x) : \Sigma_{j : J}\,\alpha_j\).

  It remains to prove that \(\hat{f}\) is a simulation. Because the defining
  properties of a simulation are propositions, we can use set quotient induction
  and it suffices to prove the following two things:

  \begin{enumerate}[(I)]
  \item If \(\alpha_i \initseg x \prec \alpha_j \initseg y\), then
    \(b_i^x \prec b_j^y\).
  \item If \(b \prec b_i^x\), then there exists \(j : I\) and \(y : \alpha_j\) such
    that \(\alpha_i \initseg y \prec \alpha_j \initseg x\) and \(b_j^y = b\).
  \end{enumerate}
  For (I), observe that if \(\alpha_i \initseg x \prec \alpha_j \initseg y\),
  then \(\beta \initseg b_i^x \prec \beta \initseg b_j^y\), from which
  \(b_i^x \prec b_j^y\) follows using~\cref{initseg-order}.
  For (II), suppose that \(b \prec b_i^x\). Because \(f_i\) (see
  item~\ref{assum2} above) is a simulation, there exists \(y : \alpha_i\) with
  \(y \prec x\) and \(f_i(y) = b\). By~\cref{initseg-order}, we get
  \(\alpha_i \initseg y \prec \alpha_i \initseg x\). Moreover,
  \(b_i^y = f_i(y) = b\), finishing the proof of~(II).
\end{proof}

In~\cref{sec:set-replacement} we saw that set replacement is equivalent to the
existence of small set quotients, so the following result immediately follows
from the theorem above. But the point is that an alternative construction
without set quotients is available, if set replacement is assumed.

\begin{thm}
  Assuming set replacement, the large ordinal \(\Ord_{\V}\) has suprema of
  families indexed by types in \(\V\).
\end{thm}

\begin{proof}[Proof~{\cite{TypeTopologyOrdinalSuprema}}]
  Given \(\alpha : I \to \Ord_{\V}\), consider the image of the map
  \(e : \Sigma_{i : I}\,\alpha_i \to \Ord_{\V}\) given by
  \(e(i,x) \colonequiv \alpha_i \initseg x\).
  The image of \(e\) is equivalent to the type
  \(\Sigma_{\gamma : \Ord_{\V}}\exists_{i : I}\, \gamma \prec \alpha_i\), \ie\
  the type of ordinals that are initial segments of some \(\alpha_i\).
  One can prove that \(\image(e)\) with the induced order from \(\Ord_{\V}\) is
  again a well-order and that for every \(i : I\), the canonical map
  \(\alpha_i \to \image(e)\) is a simulation.
  Moreover, if \(\beta\) is an ordinal such that \(\alpha_i \preceq \beta\) for
  every \(i : I\), then for every \(i : I\) and every \(x : \alpha_i\) there
  exists a unique \(b_i^x : \beta\) such that
  \(\alpha_i \initseg x = \beta \initseg b_i^x\).
  Now observe that for every \(\gamma : \Ord_{\V}\), the map
  \(\pa*{\Sigma_{i : I}\Sigma_{x : \alpha_i}\,\pa*{\alpha_i \initseg x =
      \gamma}} \to \beta\) defined by the assignment
  \((i , x , p) \mapsto b_i^x\) is a constant function to a set.
  Hence, by \cite[Theorem~5.4]{KrausEtAl2017}, this map factors through the
  propositional truncation
  \(\exists_{i : I}\Sigma_{x : \alpha_i}\,\pa*{\alpha \initseg x = \gamma}\).
  This yields a map \(\image(e) \to \beta\) which can be proved to be a
  simulation, as desired.
  Finally, we use set replacement and the fact that \(\Ord_{\V}\) is locally
  \(\V\)-small (by univalence) to get an ordinal in \(\V\) equivalent to
  \(\image(e)\), finishing the proof.
\end{proof}

\section{Families and Subsets}
\label{sec:families-and-subsets}
In traditional impredicative foundations, completeness of posets is usually
formulated using subsets. For instance, dcpos are defined as posets \(D\) such
that every directed subset of \(D\) has a supremum in
\(D\). \cref{examples-of-delta-complete-posets} are all formulated using small
families instead of subsets. While subsets are primitive in set theory, families
are primitive in type theory, so this could be an argument for using families
above. However, that still leaves the natural question of how the family-based
definitions compare to the usual subset-based definitions, especially in our
predicative setting, unanswered. This section addresses this question. We first
study the relation between subsets and families predicatively and then clarify
our definitions in the presence of impredicativity.
In our answers we will consider sup-lattices, but similar arguments could be
made for posets with other sorts of completeness, such as dcpos.

We first show that simply asking for completeness with respect to all subsets is
not satisfactory from a predicative viewpoint. In fact, we will now see that
even asking for completeness with respect to all elements of
\(\powerset_{\T}(X)\) for some fixed universe \(\T\) is problematic from a
predicative standpoint, where we recall that
\(\powerset_{\T}(X) \equiv (X \to \Omega_{\T})\).

\begin{defi}[\(\T\)-valued subsets]
  For a universe \(\T\) and a type \(X\) in any universe, we refer to the
  elements of \(\powerset_{\T}(X)\) as \emph{\(\T\)-valued subsets} of \(X\).
\end{defi}

\begin{thm}
  \label{all-T-subsets-resizing}
  Let \(\U\) and \(\V\) be universes, fix a proposition \(P_{\U} : \U\) and
  recall \(\lifting_{\V}(P_{\U})\) from~\cref{lifting-of-prop}, which has
  \(\V\)-suprema.
  If \(\lifting_{\V}(P_{\U})\) has suprema for all \(\T\)-valued subsets, then
  \(P_{\U}\) is \(\V\)-small independently of the choice of the type universe
  \(\T\).
\end{thm}
\begin{proof}
  Let \(\T\) be a type universe and consider the subset \(S\) of
  \(\lifting_{\V}(P_\U)\) given by \(Q \mapsto \One_\T\).  Note that \(S\) has a
  supremum in \(\lifting_{\V}(P_\U)\) if and only if \(\lifting_{\V}(P_\U)\) has
  a greatest element, but by~\cref{maximal-iff-resize}, the latter is equivalent
  to \(P_\U\) being \(\V\)-small.
\end{proof}

The proof above illustrates that if we have a subset \(S : \powerset_{\T}(X)\),
then there is no reason why the total space \({\Sigma_{x : X} (x \in S)}\)
should be \(\T\)-small. In fact, for \(S(x) \colonequiv \One_{\T}\) as above,
the latter is equivalent to asking that \(X\) is \(\T\)-small.

\begin{defi}[Total space of a subset, \(\totalspace\)]
  The \emph{total space} of a \(\T\)-valued subset \(S\) of a type \(X\) is
  defined as \(\totalspace(S) \colonequiv \Sigma_{x : X} (x \in S)\).
\end{defi}

In an attempt to solve the problem described in \cref{all-T-subsets-resizing},
we look to impose size restrictions on the total space of a subset. There are
two natural such restrictions and they are reminiscent of Bishop and Kuratowski
finite subsets.

\begin{defi}[\(\V\)-small and \(\V\)-covered subsets]
  An element \(S : \powerset_{\T}(X)\) is
  \begin{enumerate}[(i)]
  \item \emph{\(\V\)-small} if its total space is \(\V\)-small, and
  \item \emph{\(\V\)-covered} if we have \(I : \V\) with a surjection
    \(e : I \surj \totalspace(S)\). %
    \qedhere
  \end{enumerate}
\end{defi}

Observe that every \(\V\)-small subset is \(\V\)-covered, because every
equivalence is a surjection.
But the converse does not hold: We can emulate the well-known argument used to
show that, constructively, Kuratowski finiteness does not necessarily imply
Bishop finiteness to show that, predicatively, being \(\V\)-covered does not
necessarily imply being \(\V\)-small.

\begin{prop}\label{covered-small-resizing}
  For every two universes \(\U\) and \(\V\), if every \(\V\)-covered element of
  \(\powerset_{\U}\pa*{\Omega_{\U}}\) is \(\V\)-small, then
  \(\Propresizing{\U}{\V}\) holds.
\end{prop}
\begin{proof}
  Suppose that every \(\V\)-covered \(\U\)-valued subset of \(\Omega_{\U}\) is
  \(\V\)-small and let \(P : \U\) be an arbitrary proposition. Consider the
  subset \(S_P : \Omega_{\U} \to \Omega_{\U}\) given by
  \(S_P(Q) \colonequiv \pa*{Q = P} \vee \pa*{Q = \One_{\U}}\). Notice that this
  is \(\V\)-covered as witnessed by
  \begin{align*}
    \pa*{\One_{\V} + \One_{\V}} &\surj \totalspace(S_P) \\
    \inl(\star) &\mapsto \pa*{P\hspace{1.25mm},\tosquash*{\inl(\refl)}} \\
    \inr(\star) &\mapsto \pa*{\One_{\U},\tosquash*{\inr(\refl)}},
  \end{align*}
  so by assumption \(\totalspace(S_P)\) is \(\V\)-small. But observe that \(P\)
  holds if and only if \(\totalspace(S_P)\) is a subsingleton, but the latter
  type is \(\V\)-small by assumption, hence so is \(P\).
\end{proof}

In the case where we restrict our attention to a single universe \(\V\) and a
locally \(\V\)-small set \(X\), the two notions coincide if and only if we have
set replacement for maps into \(X\) with \(\V\)-small domain.
\begin{prop}\label{covered-small-set-replacement}
  If \(X\) is locally \(\V\)-small set, then every \(\V\)-covered element of
  \(\powerset_{\V}(X)\) is \(\V\)-small if and only if the image of any map into
  \(X\) with \(\V\)-small domain is \(\V\)-small.
\end{prop}
\begin{proof}
  Suppose first that every \(\V\)-covered subset \(S : X \to \Omega_{\V}\) is
  \(\V\)-small and let \(f : I \to X\) be map such that \(I\) is \(\V\)-small.
  Without loss of generality, we may assume that \(I : \V\), because we can
  always precompose \(f\) with the equivalence witnessing that \(I\) is
  \(\V\)-small.
  Now consider the subset \(S : X \to \Omega_{\V}\) given by
  \(S(x) \colonequiv \exists_{i : I}\pa*{f(i) =_{\V} x}\), where \({=_{\V}}\)
  has values in \(\V\) and is provided by our assumption that \(X\) is locally
  \(\V\)-small.
  Then \(S\) is \(\V\)-covered, because we have
  \(I \surj \image(f) \simeq \totalspace(S)\), where the first map is the
  corestriction of \(f\).
  So by assumption \(\totalspace(S)\) is \(\V\)-small, which means that
  \(\image(f)\) must be \(\V\)-small too.

  Conversely, assume the set replacement principle and let
  \(S : X \to \Omega_{\V}\) be \(\V\)-covered by \(e : I \surj \totalspace(S)\).
  Define the subset \(S' : X \to \Omega_{\V}\) by
  \(S'(x) \colonequiv \exists_{i : I}\pa*{x =_{\V} \fst(e_i)}\).
  By the assumed set replacement principle for \(X\), the subset \(S'\) is a
  \(\V\)-small since \(\totalspace(S') \simeq \image({\fst} \circ {e})\).
  Finally, it follows from the surjectivity of \(e\) that \(S\) and \(S'\) are
  equal as subsets, and therefore that
  \(\totalspace(S) \simeq \totalspace(S')\). Hence, \(S\) is a \(\V\)-small
  subset, as desired.
\end{proof}

So, predicatively, and in the absence of a set replacement principle, the notion
of a \(\V\)\nobreakdash-small subset is strictly stronger than that of a
\(\V\)-covered subset. Hence, in this setting, having suprema for all
\(\V\)-small subsets is strictly weaker than having suprema for all
\(\V\)\nobreakdash-covered subsets.
Meanwhile, \cref{family-subset-sup-equiv} will imply that there are plenty of
examples of posets with suprema for all \(\V\)-covered subsets, even
predicatively.
So we prefer the stronger, but predicatively reasonable requirement of asking
for suprema of all \(\V\)-covered subsets.

Form a practical viewpoint, \(\V\)-covered subsets also give us an easy handle
on examples like the following: Let \(X\) be a poset with suprema for all
(directed) \(\U_0\)-covered subsets. Then the least fixed point of a Scott
continuous endofunction \(f\) on \(X\) can be computed as the supremum of the
subset \(\{\bot, f(\bot) , f^2(\bot) , \dots\}\), which is covered by \(\Nat\).
But it is not clear that this subset is \(\U_0\)-small, at least not in the
absence of set replacement.

Our preference for \(\V\)-covered subsets over \(\V\)-small subsets also makes
it clear why we do not impose an injectivity condition on families, because for
every type \(X : \U\) there is an equivalence between embeddings
\(I \hookrightarrow X\) with \(I : \V\) and \(\pa*{\U \sqcup \V}\)-valued
subsets of \(X\) whose total spaces are \(\V\)-small,
\cf~\cite[\mkTTurllong{Slice.Slice}{\%F0\%9D\%93\%95-equiv}]{TypeTopology}.

\begin{thm}
  \label{family-subset-equiv}
  For \(X : \U\) and any universe \(\V\) we have an equivalence between
  \(\V\)-covered \(\pa*{\U \sqcup \V}\)-valued subsets of \(X\) and families
  \(I \to X\) with \(I : \V\).
\end{thm}
\begin{proof}
  The forward map \(\varphi\) is given by \((S,I,e) \mapsto (I,{\fst} \circ {e})\).
  In the other direction, we define \(\psi\) by mapping \((I,\alpha)\) to the
  triple \((S,I,e)\) where \(S\) is the subset of \(X\) given by
  \(S(x) \colonequiv \exists_{i : I}\,(x = \alpha(i))\) and
  \(e : I \surj \totalspace(S)\) is defined as
  \(e (i) \colonequiv \pa*{\alpha(i),\tosquash*{(i,\refl)}}\).
  The composite \(\varphi \circ \psi\) is easily seen to be equal to the
  identity. To show that \(\psi \circ \varphi\) equals the identity, we need the
  following intermediate result, which is proved using function extensionality
  and path induction.
  \begin{clmnn}
    Let \(S,S' : X \to \Omega_{\U\sqcup\V}\), \(e : I \to \totalspace(S)\)
    and \(e' : I \to \totalspace(S')\). If \(S = S'\) and \({{\fst} \circ e \sim
    {\fst} \circ e'}\), then \((S,e) = (S',e')\).
  \end{clmnn}
  The result follows from the claim using function and propositional
  extensionality.
\end{proof}

\begin{cor}
  \label{family-subset-sup-equiv}
  A poset with carrier in \(\U\) has suprema for all \(\V\)-covered
  \(\pa*{\U \sqcup \V}\)\nobreakdash-valued subsets if and only if it has
  suprema for all families indexed by types in \(\V\).
\end{cor}
\begin{proof}
  This is because the supremum of a \(\V\)-covered subset equals the supremum of
  the corresponding family and vice versa by inspecting the proof
  of~\cref{family-subset-equiv}.
\end{proof}

We conclude by comparing our family-based approach to the subset-based approach
in the presence of impredicativity.

\begin{thm}\label{impred-comparison}
  Assuming \(\Omegaresizing{\T}{\U_0}\) for every universe \(\T\), the following
  are equivalent for a poset with carrier in a universe \(\U\):
  \begin{enumerate}[(i)]
  \item\label{all-subsets} the poset has suprema for all subsets;
  \item\label{all-covered-subsets} the poset has suprema for all \(\U\)-covered
    subsets;
  \item\label{all-small-subsets} the poset has suprema for all \(\U\)-small
    subsets;
  \item\label{all-small-families} the poset has suprema for all families indexed
    by types in \(\U\).
  \end{enumerate}
\end{thm}
\begin{proof}
  Clearly
  \ref{all-subsets}~\(\Rightarrow\)~\ref{all-covered-subsets}~\(\Rightarrow\)~\ref{all-small-subsets}.
  We show that \ref{all-small-subsets} implies \ref{all-subsets}, which proves
  the equivalence of \ref{all-subsets}--\ref{all-small-subsets}. Assume that a
  poset \(X\) has suprema for all \(\U\)-small subsets and let
  \(S : X \to \Omega_{\T}\) be any subset of \(X\). Using
  \(\Omegaresizing{\T}{\U_0}\), the total space \(\totalspace(S)\) is
  \(\U\)-small. So~\(X\)~has a supremum for \(S\) by assumption, as
  desired. Finally, \ref{all-covered-subsets}~and~\ref{all-small-families} are
  equivalent in the presence of \(\Omegaresizing{\T}{\U_0}\)
  by~\cref{family-subset-sup-equiv}.
\end{proof}
If condition~\ref{all-small-families} of~\cref{impred-comparison} holds, then
the poset has suprema for all families indexed by types in \(\V\) provided that
\(\V \sqcup \U \equiv \U\).
Typically, in the examples of~\cite{deJongEscardo2021}~for~instance,
\(\U \colonequiv \U_1\) and \(\V \colonequiv \U_0\), so that
\(\V\sqcup\U \equiv \U\) holds. Thus, our \(\V\)-families-based approach
generalizes the traditional subset-based approach.

\section{Conclusion}
\label{sec:conclusion}
Firstly, we have shown, constructively and predicatively, that nontrivial dcpos,
bounded complete posets and sup-lattices are all necessarily large and
necessarily lack decidable equality. We did so by deriving a weak
impredicativity axiom or weak excluded middle from the assumption that such
nontrivial structures are small or have decidable equality,
respectively. Strengthening nontriviality to the (classically equivalent)
positivity condition, we derived a strong impredicativity axiom and full
excluded middle.

Secondly, we showed that Tarski's greatest fixed point theorem cannot be applied
in nontrivial instances in our predicative setting, while generalizations of
Tarski's theorem that allow for large structures are provably
false. Specifically, we showed that the ordinal of ordinals in a univalent
universe does not have a maximal element, but does have small suprema in the
presence of small set quotients, or equivalently, set replacement.
More generally, we investigated the inter-definability and interaction of type
universes of propositional truncations and set quotients in the absence of
propositional resizing axioms. In particular, we showed that in the presence of
propositional truncations, but without assuming propositional resizing, it is
possible to construct set quotients that happen to live in higher type universes
but that do satisfy the appropriate universal properties with respect to sets in
arbitrary type universes.

Finally, we clarified, in our predicative setting, the relation between the
traditional definition of a lattice that requires completeness with respect to
subsets and our definition that asks for completeness with respect to small
families.

In future work, it would be interesting to study the predicative validity of
Pataraia's theorem and Tarski's \emph{least} fixed point theorem.
Curi~\cite{Curi2015,Curi2018} develops predicative versions of Tarski's fixed
point theorem in some extensions of CZF. It is not clear whether these arguments
could be adapted to univalent foundations, because they rely on the
set-theoretical principles discussed in the introduction. The availability of
such fixed-point theorems might be useful for application to inductive
sets~\cite{Aczel1977}, which we might otherwise introduce in univalent
foundations using higher inductive types~\cite{HoTTBook}.
In another direction, we have developed a notion of
apartness~\cite{BridgesVita2011} for continuous dcpos~\cite{GierzEtAl2003} that
is related to the strictly-below relation introduced in this paper. Namely, if
\(x \below y\) are elements of a continuous dcpo, then \(x\) is strictly below
\(y\) if \(x\) is apart from \(y\). In \cite{deJong2021}, we give a constructive
analysis of the Scott topology~\cite{GierzEtAl2003} using this notion of
apartness.

\section{Acknowledgements}
We would like to thank the reviewers for their valuable and complementary
suggestions. We~are particularly grateful to the reviewer who pointed out that
one of our results can be strengthened to~\cref{is-small-retract} and for their
insights and questions
on~\cref{sec:maximal-and-fixed-points,sec:families-and-subsets} that have
considerably improved the paper.

\bibliographystyle{alphaurl}
\bibliography{references.bib}

\end{document}

%% file: mymacros.tex
\usepackage{mathtools}

\newcommand{\colonequiv}{\mathrel{\vcentcolon\mspace{-1.2mu}\equiv}}
\DeclarePairedDelimiter{\pa}{(}{)}

\newcommand{\Zero}{\mathsf{0}}
\newcommand{\Nat}{\mathbb{N}}
\newcommand{\One}{\mathsf{1}}
\newcommand{\Two}{\mathsf{2}}
\DeclareMathOperator{\inl}{inl}
\DeclareMathOperator{\inr}{inr}

\DeclareMathOperator{\issubsingleton}{is-subsingleton}
\DeclareMathOperator{\issingleton}{is-singleton}
\DeclarePairedDelimiter{\squash}{\|}{\|}
\DeclarePairedDelimiter{\tosquash}{|}{|}

\DeclarePairedDelimiter{\squashVV}{\|}{\|_{v}}
\DeclarePairedDelimiter{\tosquashVV}{|}{|_v}

\DeclareMathOperator{\id}{id}
\DeclareMathOperator{\Id}{Id}
\DeclareMathOperator{\refl}{refl}

\DeclareMathOperator{\fib}{fib}
\DeclareMathOperator{\image}{im}
\newcommand{\surj}{\twoheadrightarrow}

\DeclareMathOperator{\totalspace}{\mathbb T}

\DeclareMathOperator{\fst}{pr_1}
\DeclareMathOperator{\snd}{pr_2}

\newcommand{\U}{\mathcal U}
\newcommand{\V}{\mathcal V}
\newcommand{\W}{\mathcal W}
\newcommand{\T}{\mathcal T}

\newcommand{\issmall}[1]{\operatorname{is}\;\!\mathcal{#1}\!\operatorname{-small}}

\DeclareMathOperator{\propresizing}{Propositional-Resizing}
\newcommand{\Propresizing}[2]{\propresizing_{#1,#2}}
\DeclareMathOperator{\resizing}{-Resizing}
\newcommand{\Omegaresizing}[2]{\Omega\!\resizing_{#1,#2}}
\newcommand{\Omegaresizingalt}[1]{\Omega\!\resizing_{#1}}

\newcommand{\Omeganotnot}[1]{\Omega^{\lnot\lnot}_{#1}}
\newcommand{\Omeganotnotresizing}[2]{{\Omega_{\lnot\lnot}}\!\resizing_{#1,#2}}
\newcommand{\Omeganotnotresizingalt}[1]{{\Omega_{\lnot\lnot}}\!\resizing_{#1}}

\DeclareMathOperator{\tarski}{Tarski's-Theorem}
\newcommand{\Tarski}[3]{\tarski_{#1,#2,#3}}

\DeclareMathOperator{\lifting}{\mathcal L}

\newcommand{\below}{\mathrel\sqsubseteq}
\newcommand{\sbelow}{\mathrel\sqsubset}
\newcommand{\aboveorder}{\mathrel\sqsupseteq}
\DeclareMathOperator{\powerset}{\mathcal P}
\newcommand{\deltacomplete}[1]{\(\delta_{#1}\)-complete}

\DeclareMathOperator{\Ord}{Ord}
\newcommand{\initseg}{\mathbin\downarrow}

\newcommand{\mkTTurl}[1]{\href{https://www.cs.bham.ac.uk/~mhe/TypeTopology/Published.#1.html}{\texttt{\textup{#1}}}}
\newcommand{\mkTTurllong}[2]{\href{https://www.cs.bham.ac.uk/~mhe/TypeTopology/Published.#1.html\##2}{\texttt{\textup{#1}}}}